\documentclass[a4paper, 11pt]{article}
\usepackage[english]{babel}
\usepackage{amsmath,  amssymb, amsfonts}
\usepackage{amsthm}
\usepackage{mathrsfs}
\usepackage{color}
\usepackage{graphicx}
\usepackage{enumerate}
\usepackage{bm}
\usepackage{cite}
\usepackage{url}
\usepackage{algorithm} 
\usepackage{algorithmic} 
\usepackage{subfigure}
\usepackage{mathtools}
\usepackage[mathcal]{euscript}

\newcommand{\R}{\mathbb{R}}
\newcommand{\Z}{\mathbb{Z}}

\newcommand{\1}{\mathbf{1}}
\DeclareMathOperator{\diag}{diag}
\DeclareMathOperator{\rank}{rank}
\DeclareMathOperator{\spa}{span}

\DeclareMathOperator{\tr}{tr}

\DeclareMathOperator{\range}{range}

\newcommand{\bb}{\mathbf{b}}

\newcommand{\ub}{\mathbf{u}}
\newcommand{\vb}{\mathbf{v}}

\newcommand{\xb}{\mathbf{x}}
\newcommand{\yb}{\mathbf{y}}
\newcommand{\zb}{\mathbf{z}}

\newcommand{\Hb}{\mathbf{H}}

\newcommand{\Ib}{\mathbf{I}}
\newcommand{\Mb}{\mathbf{M}}
\newcommand{\omegab}{\bm{\omega}}

\newcommand{\Qb}{\mathbf{Q}}
\newcommand{\Fb}{\mathbf{F}}
\newcommand{\Eb}{\mathbf{E}}

\newcommand{\Ab}{\mathbf{A}}
\newcommand{\Bb}{\mathbf{B}}
\newcommand{\Cb}{\mathbf{C}}

\newcommand{\Tb}{\mathbf{T}}

\newcommand{\rb}{\mathbf{r}}

\newcommand{\Gb}{\mathbf{G}}

\newcommand{\mV}{\mathrm{V}}

\newcommand{\gammab}{{\bm \gamma}}

\newcommand{\betab}{\bm{\beta}}

\newcommand{\Sigmab}{\bm{\Sigma}}

\DeclareFontFamily{OT1}{pzc}{}
\DeclareFontShape{OT1}{pzc}{m}{it}{<-> s * [1.200] pzcmi7t}{}
\DeclareMathAlphabet{\mathpzc}{OT1}{pzc}{m}{it}

\newcommand*\mcap{\mathbin{\mathpalette\mcapinn\relax}}
\newcommand*\mcapinn[2]{\vcenter{\hbox{$\mathsurround=0pt
			\ifx\displaystyle#1\textstyle\else#1\fi\bigcap$}}}

\newcommand*\mcup{\mathbin{\mathpalette\mcupinn\relax}}
\newcommand*\mcupinn[2]{\vcenter{\hbox{$\mathsurround=0pt
			\ifx\displaystyle#1\textstyle\else#1\fi\bigcup$}}}

\newcommand{\mR}{\mathrm{R}}

\newcommand{\dH}{\delta_{\rm H}}
\newcommand{\dz}{\delta_{\rm z}}
\newcommand{\mN}{\mathrm{N}}

\newcommand{\mM}{\mathscr{M}}
\DeclareMathOperator{\Lap}{Lap}
\newcommand{\Lapn}{\Lap^n}
\newcommand{\Lapm}{\Lap^m}
\newcommand{\xbs}{\xb^\sharp}
\newcommand{\xbf}{\xb^\flat}
\newcommand{\Exp}{\mathbb{E}}

\newcommand{\pdf}{\mathrm{pdf}}
\newcommand{\eq}{\mathcal{E}}

\newcommand{\Eq}{\mathscr{E}}

\newcommand{\iast}{i^\ast}
\newcommand{\Wb}{\mathbf{W}}

\newcommand{\ab}{\mathbf{a}}
\newcommand{\PP}{\mathpzc{P}}
\newcommand{\PPo}{\mathpzc{P}_{\Omega}}
\newcommand{\dA}{\delta_{\rm A}}
\newcommand{\db}{\delta_{\rm b}}
\newcommand{\Ub}{\mathbf{U}}
\newcommand{\Rb}{\mathbf{R}}
\newcommand{\Sb}{\mathbf{S}}
\newcommand{\Vb}{\mathbf{V}}

\newcommand{\sigm}{\sigma_{\rm m}}
\newcommand{\sigM}{\sigma_{\rm M}}

\newcommand{\tH}{\mathcal{Z}_H}
\newcommand{\tz}{\mathbf{z}_{\mathbf{H}}}

\newcommand{\mO}{\mathcal{O}}

\newcommand{\Yb}{\mathbf{Y}}
\newcommand{\oY}{\overline{\Yb}}
\newcommand{\uY}{\underline{\Yb}}

\newcommand{\lam}{\lambda_{\rm m}}
\newcommand{\laM}{\lambda_{\rm M}}

\newtheorem{proposition}{Proposition}
\newtheorem{theorem}{Theorem}
\newtheorem{lemma}{Lemma}

\newtheorem{definition}{Definition}

\title{\bf Dynamical Privacy in  Distributed Computing\\
	Part I:  Privacy Loss and PPSC Mechanism}
\date{}
\author{Yang Liu, Junfeng Wu, Ian Manchester, Guodong Shi\thanks{Y. Liu is with the Research School of Engineering, The Australian National University, Canberra 0200, Australia. (email: yang.liu@anu.edu.au).}%
\thanks{J. Wu is with   with the College of Control Science and Engineering, Zhejiang
University, Hangzhou 310027, China (e-mail: jfwu@zju.edu.cn)}%
\thanks{I. R. Manchester is with Australian Center for Field Robotics, The University of Sydney, NSW 2006, Australia. (email: ian.manchester@sydney.edu.au)}
\thanks{G. Shi is with Australian Center for Field Robotics, The University of Sydney, NSW 2006, Australia, and the Research School of Engineering, The Australian National University, Canberra 0200, Australia (email: guodong.shi@sydney.edu.au)}
 }

\begin{document}
\maketitle

\begin{abstract}
	A distributed computing protocol consists of three components: (i) Data Localization:  a network-wide dataset is decomposed into  local datasets separately preserved at a network of nodes; (ii) Node Communication:  the nodes hold individual dynamical states and communicate with the neighbors about these dynamical states; (iii) Local Computation: state recursions are computed at each individual node. Information about the local datasets enters the computation process through the node-to-node communication and  the local computations, which may be leaked to dynamics eavesdroppers having access to global or local node states.  In this paper, we systematically investigate this potential  computational privacy risks in distributed computing protocols in the form of structured system identification, and then propose and thoroughly analyze a Privacy-Preserving-Summation-Consistent (PPSC) mechanism as a generic privacy encryption subroutine for  consensus-based distributed computations. The central idea is  that the consensus manifold is where we can both hide node privacy and achieve  computational accuracy.  In this first part of the paper, we demonstrate the computational privacy risks in distributed algorithms against dynamics eavesdroppers and particularly in distributed linear equation solvers, and then propose the
	PPSC mechanism and illustrate its usefulness.
\end{abstract}

\section{Introduction}

The development of distributed control and optimization  has become one of the central  streams in the study of complex network operations,  due to the rapidly growing volume and dimension of data and information flow from a variety of applications  such as social networking, smart grid, and intelligent transportation \cite{magnusbook}. Celebrated results have been established for problems ranging from optimization  \cite{nedic09,nedic10} and learning  \cite{boyd2011distributed} to formation \cite{xiaoming-formation} and localization  \cite{khan2009distributed}, where distributed control and decision-making solutions provide resilience and scalability for large-scale problems through allocating information sensing and decision making over individual agents. The origin of this line of research can be traced back to the 1980s from the work of distributed optimization and decision making \cite{tsi} and parallel computations \cite{lynch}.
Consensus algorithms serve as a basic tool for the information dissemination of distributed algorithms \cite{jad03,xiao04}, where the goal is to
drive the node states of a network to a common value related to the network initials by node interactions respecting the network structure, usually just the average.

The primary power of consensus algorithms lies in  their distributed nature in the sense that throughout the recursions of the algorithms, individual nodes only share their state information with a few other nodes that are connected to  or trusted which are called neighbors.   As a result, building on consensus algorithms, many algorithms can be designed which utilize the   consensus mechanism  to achieve network-level control and computation objectives, e.g., the recent work on network linear equation solvers \cite{Mou-TAC-2015,Shi-TAC-2017}, and convex optimization of a sum of  local objective functions \cite{nedic10}. This emerging progress on distributed computation is certainly an extension of the classical development of distributed and parallel computation, but the new emphasis  is on how a distributed computing algorithm for a network-wide problem can be designed with robustness and resilience over a given communication structure \cite{Mou-TAC-2015,Shi-TAC-2017}, instead of breaking down a high-dimensional system to a number of computing units with all-to-all communications.

Under a distributed computing structure, the
information about the network-wide computation problem is encoded in the individual node initial values or update rules, which are not shared directly among the nodes. Rather, nodes share certain dynamical states based on neighboring communications and local problem datasets, in order to find solutions of the global problem. The node dynamical states may contain sensitive private information  directly or indirectly, but  in distributed computations to achieve the network-wide computation goals nodes   have to accept the fact that their privacy, or at least some of it, will be inevitably lost. It was pointed out that indeed nodes  lose their  privacy in terms of initial values if an attacker, a malicious user, or an eavesdropper knows part of the node trajectories and the network structure \cite{sundaram2007,yuan2013}. In fact, as long as observability holds, the whole network initial values become re-constructable with only a segment of node state updates at a few selected nodes.  Several insightful privacy-preserving consensus algorithms have been presented in the literature \cite{huang2012differentially,manitara2013privacy,mo2017privacy,nicolas2017,claudio2018} in the past few years, where the central idea is to inject  random noise or offsets in the node communication and iterations, or employ local maps as node state masks. The classical notion on differential privacy in computer science has also been introduced to  distributed optimization problems   \cite{huang2015differentially,pappas2017}, following the earlier work on differentially private filtering \cite{pappas2014}.

In this paper, we systematically investigate the potential  computational privacy risks in distributed computing protocols, and then propose and thoroughly analyze a Privacy-Preserving-Summation-Consistent (PPSC) mechanism as a generic privacy encryption step for consensus-based distributed computations.   In  distributed computing protocols, a network-wide dataset is decomposed into  local datasets which are separately  located at a network of nodes; the nodes hold individual dynamical states, based on which communication packets are produced and shared  according to the network structure; local computations are then carried out at each individual node. The information about the local datasets enters the computation process through the node-to-node communication or the local computations; while techniques from system identification may be able to recover such datasets by observations of the node states.  In this first part of the paper, we reveal the identification-type computational risk in distributed algorithms and particularly in distributed linear equation solvers, and then define a so-called
PPSC mechanism for privacy preservation. The main results of the paper are summarized as follows.

\begin{itemize}
	\item We show that distributed computing is exposed to privacy risks in terms the local datasets from both global and local eavesdroppers having access to the entirety or part of the node states. Particularly, for  network linear equations as a basic computation task, we show explicitly how  the computational privacy has indeed been completely lost to global eavesdroppers for almost all initial values, and partially to local passive or active eavesdroppers  who can monitor or alter selected node states.
	\item We propose a Privacy-Preserving-Summation-Consistent (PPSC) mechanism which is shown to be useful as  a generic privacy encryption step for  consensus-based distributed computations. The central idea is to make use of the consensus manifold to on one hand hide node privacy, and on the other hand ensure computation accuracy and efficiency.
\end{itemize}
We remark that the existence of lower-dimensional manifold for the network-wide computation  goal is very much  not  unique only for consensus-based distributed computing, and therefore the idea of the PPSC mechanism may be extended to other distributed computing methods. In the second part of the paper, we demonstrate that   PPSC mechanism can be realized by conventional gossip algorithms, and establish their detailed privacy-preserving capabilities  and convergence efficiencies. 

Some preliminary ideas and results of this work were presented at
IEEE Conference on Decision and Control in Dec. 2018   \cite{cdc1,cdc2}.  In the current manuscript, we have established a series of   new results on privacy loss characterizations and privacy preservation quantifications, in addition to a few new illustration  examples and technical proofs.

The remainder the paper is organized as follows. In Section \ref{sec:problem}, we introduce a general model for distributed computing protocols, based on which we define the resulting computational privacy. The connection between computational privacy and network structure and initial-value privacies is also discussed.  In Section \ref{Sec:privacy}, we investigate the potential loss of computational privacies in distributed computing algorithms, and in particular in distributed linear equation solvers. We also show that conventional methods for differential privacy often lead to non-ideal computation result for distributed computings facing fundamental  privacy vs. accuracy trade-off in differential privacy can be challenging for a recursive computing process. In Section \ref{Sec:PPSC}, we define the PPSC mechanism and illustrate how we can design privacy-preserving algorithms for average consensus, distributed linear equation solvers, and in general for distributed convex optimization based on PPSC mechanism. Section \ref{sec:local} extends the discussion on computational privacy loss to local eavesdroppers and establish some results using classical system identification methods. Finally Section \ref{sec:conc} concludes the paper with a few remarks. The various statements established in this paper is collected in the appendix.

\medskip

\noindent{\bf Notation.} We let $\Hb_{ij},\Hb_i$ denote the $ij$--th entry and the transpose of the $i$--th row of a matrix $\Hb$, respectively. For a vector $\vb$, $[\vb]_i$ denotes its $i$--th component. Let $\diag(a_1,\dots,a_n)$ with $a_i\in\R$ denote a diagonal matrix with diagonal entries being $a_1,\dots,a_n$ in order.
The transpose of a matrix or a vector $\Hb$ are denoted by $\Hb^\top$. Let $\tr(\cdot)$ denote the trace of a square matrix. The orthogonal complement of a vector space $\mathrm{S}$ is denoted by $\mathrm{S}^\perp$. Let $\otimes$ denote the Kronecker product of two matrices or the direct sum of two sets. Conventionally, $\|\cdot\|_p$ is the $p$--norm of a vector or a matrix, and $\|\cdot\|$ is its $2$--norm.
Let $\lam(\cdot),\laM(\cdot)$ be the smallest and the largest eigenvalue of a real symmetric matrix. Let $\sigM(\cdot),\sigm(\cdot)$ denote the maximum and minimum absolute value of eigenvalues of a real symmetric matrix.
Let $\|\cdot\|_{\rm F}$ denote the Frobenius norm of a matrix. The range of a matrix or a function is denoted as $\range(\cdot)$. We use $\Z,\Z^{\ge0},\Z^+,\R,\R^+$ to represent the set of integers, the set of nonnegative integers, the set of positive integers, the set of real numbers, and the set of positive real numbers, respectively. We let $\Pr(\cdot)$ be the probability of an event in some probability space. We let $\pdf(\cdot),\Exp(\cdot)$ denote the probability density function (PDF) and the expected value of a random variable, respectively. Let $\Lap(v)$ with $v>0$ denote the zero--mean Laplace distribution with variance $2v^2$. Correspondingly, $\Lapn(v)$ represents an $n$--dimensional vector of identical and independent random variables distributed according to $\Lap(v)$.

\section{Privacy Notions in Distributed Computation}\label{sec:problem}
In this section, we first introduce a basic  model for distributed computation protocols which clearly  indicates the problem  decomposition, node communication, and node state update along the recursions. Next, we discuss potential eavesdroppers for   distributed computing
protocols and the resulting categories of privacy losses at both network and individual node levels.

\subsection{Distributed Computations}
Consider a network of $n$ nodes indexed in $\mathrm{V}=\{1,\dots,n\}$ whose communication structure is described by a graph $\mathrm{G}=(\mathrm{V},\mathrm{E})$. We assume $\mathrm{G}$ is an undirected connected graph for the sake of having  a simplified presentation.  Let $\mathpzc{D}$ be a dataset whose elements are embedded in different dimensional Euclidean spaces. We write $\mathpzc{D}=\mathpzc{D}^\prime$ if $\mathpzc{D}$ and $\mathpzc{D}^\prime$ have identical records (up to some equivalence relation). Suppose $\mathpzc{D}$ is allocated over the network by partition $\mathpzc{D}= \mathpzc{D}_1 \otimes \cdots \otimes \mathpzc{D}_n$, where node $i$ holds   $\mathpzc{D}_i$ and $\otimes$ represents direct sum  to indicate potential logical or computational  structure of the data points $\mathpzc{D}$. Each node holds a dynamical sate $\mathbf{x}_i(t) \in \mathbb{R}^m$.  Let $\mathrm{N}_i=\big\{j:\{i,j\}\in \mathrm{E} \big\}$ be the neighbor set of node $i$. The following protocol describes a general procedure for distributed computing algorithms.

\begin{algorithm}[htb]{{\bf DCP} Distributed Computing Protocols}
	\begin{algorithmic}[1]
		\STATE Set $t\gets0$ and $\mathbf{x}_i(0)\in \mathbb{R}^m$ for $i\in\mathrm{V}$.
		\STATE Each node $i$ computes $\mathbf{c}_i(t)=\mathpzc{l}_{i,t}\big(\mathbf{x}_i(t),\mathpzc{D}_i\big)$ and sends $\mathbf{c}_i(t)$ to all its neighbors $j\in\mN_i$.
		\STATE Each node $i $ updates their states  $\xb_i(t+1) \gets \mathpzc{f}_{i,t}\Big( \xb_i(t);\  \mathbf{c}_j(t),j\in\mN_i;\ \mathpzc{D}_i\Big) $.
		\STATE Set $t\gets t+1$ and go to Step 2.
	\end{algorithmic}
\end{algorithm}

The $\mathbf{c}_i(t)$ are the intermediate quantities for communications among the nodes, which often coincide with the nodes states $\mathbf{x}_i(t)$. The mappings $\mathpzc{l}_{i,t}$ and $\mathpzc{f}_{i,t}$ can be deterministic or random, and they are distributed in the sense that they are mappings of the local data $\mathpzc{D}_i$ and information received from neighbors.  The closed-loop of the above distributed computing protocol is described by
\begin{align}\label{eq:dcp}
\xb_i(t+1) = \mathpzc{f}_{i,t}\Big( \xb_i(t);\  \mathpzc{l}_{j,t}\big(\mathbf{x}_j(t),\mathpzc{D}_j\big),j\in\mN_i;\ \mathpzc{D}_i\Big).
\end{align}
In the following, we present a few examples in the literature showing how this DCP model covers average consensus  and consensus-based distributed computing algorithms. To this end, let $\mathbf{W}=[w_{ij}]$ be a doubly stochastic matrix in $\mathbb{R}^{n\times n}$ that complies with the graph $\mathrm{G}$, i.e., $w_{ij}>0$ if and only if $j\in \mathrm{N}_i$ over the graph $\mathrm{G}$ for $i\neq j$. In particular, we assume $w_{ii}>0$ for all $i\in\mathrm{N}_i$.

\subsubsection{Average Consensus}Each node $i$ is assigned $\mathpzc{D}_i=\beta_i \in \mathbb{R}$. The initial state $\mathbf{x}_i(0)$ is set as $\beta_i$. The mapping $\mathpzc{l}_{i,t}$ is identity, and therefore $\mathbf{c}_i(t)=\mathbf{x}_i(t)$. Then
\begin{equation}
\begin{aligned}\label{eq:ave}
\xb_i(t+1)=\mathpzc{f}_{i,t}\Big( \xb_i(t);\  \mathbf{c}_j(t),j\in\mN_i;\ \mathpzc{D}_i\Big)=
\sum_{j\in \mathrm{N}_i \mcup \{i\}} w_{ij}\xb_j(t)
\end{aligned}
\end{equation}
represents a standard average consensus algorithm.

\subsubsection{Distributed Linear Equations} Each node $i$ holds $\mathpzc{D}_i=(\mathbf{H}_i,\zb_i)$ with $\Hb_i \in\mathbb{R}^m$ and $\zb_i\in \mathbb{R}$ defining a linear equation
\begin{align}
\mathcal{E}_i: \quad \Hb_i^\top \yb= \zb_i
\end{align}with respect to an unknown $\yb\in\R^m$. The overall dataset $\mathpzc{D}$ forms the following linear algebraic equation
\begin{equation}\label{eq:LAE}
\mathcal{E}: \quad \Hb\yb=\zb
\end{equation}
with  $\Hb\in\R^{n\times m},\ \zb\in\R^n$, where $\Hb_i^\top$ is the $i$-th row of $\Hb$ and $\zb_i$ denotes the $i$-th component of $\zb$. Let $\PP_i:\R^m\to\R^m$ be the orthogonal projection onto the subspace $\big\{\yb\in\R^m:\Hb_i^\top\yb=\zb_i\big\}$, which encodes the information of $\mathpzc{D}_i$. The following two algorithms are common distributed linear equation solvers.
\begin{itemize}
	\item {[Consensus+Projection Algorithm (CPA)]} Let $\mathbf{c}_i(t)=\mathbf{x}_i(t)$. The
	following algorithm is an Euler approximation of the so--called ``consensus + projection" flow \cite{Shi-TAC-2017,liu2013asynchronous,Mou-TAC-2015}:
	\begin{equation}\label{eq:shi2017}
	\begin{aligned}
		\xb_i(t+1)=\mathpzc{f}_{i,t}\Big( \xb_i(t);\  \mathbf{c}_j(t),j\in\mN_i;\ \mathpzc{D}_i\Big) =
		\sum\limits_{j\in\mN_i\mcup\{i\}}w_{ij}\xb_j(t)+\alpha\Big(\PP_i\big(\xb_i(t)\big)-\xb_i(t)\Big),
	\end{aligned}
	\end{equation}
	where $\alpha>0$ is a step-size parameter.
	
	\item {[Projection Consensus Algorithm (PCA)]} Let $\mathbf{c}_i(t)=\PP_i\big(\mathbf{x}_i(t)\big)$. We can also have the following algorithm as a distributed linear equation solver \cite{nedic10,anderson2015decentralized}:
	\begin{equation}\label{eq:PCA}
	\begin{aligned}
	&\quad\xb_i(t+1)=\mathpzc{f}_{i,t}\Big( \xb_i(t);\  \mathbf{c}_j(t),j\in\mN_i;\ \mathpzc{D}_i\Big)= \sum\limits_{j\in\mN_i\mcup\{i\}}w_{ij}\PP_j\big(\mathbf{x}_j(t)\big).
	\end{aligned}
	\end{equation}
\end{itemize}

\subsubsection{Distributed Convex Optimization}
Each node holds $\mathpzc{D}_i$ being    a differentiable convex function $f_i:\R^m\to\R$. The overall dataset $\mathpzc{D}= \sum\limits_{i=1}^n f_i(\cdot)$. 
Again $\mathbf{c}_i(t)=\mathbf{x}_i(t)$, and the node update is given by \cite{nedic09}
\begin{equation}
\begin{aligned}\label{eq:gradient}
\xb_i(t+1)=\mathpzc{f}_{i,t}\Big( \xb_i(t);\  \mathbf{c}_j(t),j\in\mN_i;\ \mathpzc{D}_i\Big)=\sum_{j\in \mathrm{N}_i\mcup\{i\}} w_{ij}\xb_j(t) -\epsilon_t \nabla f_i(\mathbf{x}_i(t)),
\end{aligned}
\end{equation}
which describes a distributed gradient descent algorithm.

\subsection{Communication and Dynamics  Eavesdroppers}
In practice, there  can be two typical  types of malicious eavesdroppers for a distributed computing process. Communication eavesdroppers are adversaries who  can intercept node-to-node communications, i.e., the $\mathbf{c}_i(t)$. Dynamics eavesdroppers are adversaries who can monitor the node dynamical states  $\mathbf{x}_i(t)$. For the choice of the communication mechanism
$$
\mathbf{c}_i(t)=\mathpzc{l}_{i,t}\big(\mathbf{x}_i(t),\mathpzc{D}_i\big),
$$
we can see that eavesdroppers that have access to both the node states and node communications may easily identify $\mathpzc{D}_i$ 
to the extent that $\mathbf{c}_i(t)$ depends on the data $\mathpzc{D}_i$. 
For the two categories of eavesdroppers, they can also have distinct capabilities.
\begin{itemize}
	\item[(i)] Eavesdroppers may have different knowledge about the network, e.g., knowledge of the number of nodes, network structure, and mechanisms of the communication and node updates.
	\item[(ii)] Eavesdroppers can be local or global in the sense that they are able to monitor the communication or node states of the entire network or only at one or a few nodes.
	\item[(iii)] Eavesdroppers can be active or passive by whether they are able to or willing to influence the observed data. E.g., a passive communication  eavesdropper records $\mathbf{c}_i(t)$ at a node $i$ only, while an active eavesdropper might aim to replace $\mathbf{c}_i(t)$ with some injected data.
\end{itemize}
In this paper, we focus on dynamics eavesdroppers, motivated by  applications of distributed computing where the node dynamical states $\mathbf{x}_i(t)$ might represent physical quantities, e.g. the powers and voltages of generators in a smart grid, the velocities and accelerations  of trucks in a platoon, and the workloads of computing servers in a cloud. Such physical quantities might be   potentially measurable by malicious sensors. 
We limit our discussion on the effect of  communication eavesdroppers because on the one hand, existing  cryptography methods have the capacity to enable secure node-to-node communication; and on the other hand, communication eavesdroppers fall to the category of dynamics eavesdroppers when $
\mathbf{c}_i(t)=\mathbf{x}_i(t)$, which holds true for many existing distributed computing algorithms.

\subsection{Privacy Notions and Related Works}

There can be various notions of privacies in a dynamical network process. First of all, the node dynamics reveal information about the network itself. As a result, dynamics eavesdroppers may infer statistics or even the topology of the network. This type of study has been investigated in the context network tomography and network identification problems \cite{mesbahi2012,duffield2006,materassi2012,vdh2018}.
 Furthermore, for an individual node $i$,   its dynamical state $\mathbf{x}_i(t)$, and especially its initial value $\mathbf{x}_i(0)$ may directly contain sensitive  information. We present the following illustrative example.

\medskip

\noindent{\bf Example 1}. Let $i$ represent an individual in a company of in an organization. Let $\beta_i\in\mathbb{R}$ be the monthly salary of the individual $i$. For each month, the individuals can run the average consensus algorithm (\ref{eq:ave}) with  $\mathbf{x}_i(0)=\beta_i$, and as $t$ tends to infinity, each node will be able to learn the value of $\sum_{j=1}^{n}\beta_j/n$. Along this computation  process (\ref{eq:ave}), it has been shown that with observability condition and knowledge about the weight matrix $\mathbf{W}$, it is possible to recover  $\mathbf{x}(0)$ from observation of a finite trajectory at a single node $i$ \cite{yuan2013}. On the other hand, the network structure $\mathrm{G}$ may be re-constructible  by a node knockout process established in \cite{mesbahi2012}. Additionally, in this process, node $i$ reveal its states to its neighbors, which means that it loses its privacy with respect to $\beta_i$ to neighbors right at the first step.

\medskip

In this paper, we are primarily interested in  the privacy issues defined by the local datasets $\mathpzc{D}_i$ against dynamics eavesdroppers. These local datasets $\mathpzc{D}_i$ may be sensitive in terms of privacy in general dynamical networks, e.g., the local cost functions in distributed economical MPC directly reflects the economy of individual subsystems \cite{ruigang2017}.  We aim to establish concrete results on the possibilities of identifying the local datasets $\mathpzc{D}_i$ in distributed computing schemes. This problem is precisely  a   parameter identification problem in classical system identification literature \cite{glover1976,ljung1994}, but with specific dynamics structure inherited from the network topology and problem setup. We term this type of privacy as computational privacy in a distributed computing protocol, and use  distributed   solvers  for network linear equations as a basic  example to illustrate the computational privacy loss.

\section{Dynamical Privacy and PPSC Mechanism}\label{Sec:privacy}

In this section, we define and specify computational risks in a distributed computing protocol, subject to either global or local eavesdroppers of node dynamical states.

\subsection{Dynamics Eavesdroppers}

We introduce the following definition.
\begin{definition}\label{def:global_eavesdroppers}
	For a distributed computing protocol (\ref{eq:dcp}), a global dynamics eavesdropper is a party who has access to the node states $\big(\xb(t)\big)_{t=0}^\infty$ with initial condition $\xb(0)= \xb_0$. For global dynamics eavesdroppers, the distributed computing protocol (\ref{eq:dcp}) is    computationally  private with respect to the dataset $\mathpzc{D}$ under initial condition $\xb_0\in\R^{nm}$ if the dataset $\mathpzc{D}$ is not recoverable from $\big(\xb(t)\big)_{t=0}^\infty$.
\end{definition}

A distributed computing protocol (\ref{eq:dcp}) induces a mechanism   $\mathscr{M}$ which maps from the space of the data $\mathpzc{D}$ to the space of the trajectories $\big(\xb(t)\big)_{t=0}^\infty$,   by $$
\mathscr{M}_{\xb_0}(\mathpzc{D})=\big(\xb(t)\big)_{t=0}^\infty.
$$
The following statement provides a generic condition for computational privacy.

\begin{proposition}\label{prop1}
	For the distributed computing protocol (\ref{eq:dcp}), the following statements hold.
	
	(i) With deterministic update, the distributed computing protocol (\ref{eq:dcp}) is globally  computationally  private with respect to the dataset $\mathpzc{D}$ under initial condition $\xb_0\in\R^{nm}$ if and only if there exists $\mathpzc{D}^\prime\neq\mathpzc{D}$ such that $\mathscr{M}_{\xb_0}(\mathpzc{D})=\mathscr{M}_{\xb_0}(\mathpzc{D}^\prime)$.
	
	(ii) With randomized update, the distributed computing protocol (\ref{eq:dcp}) is globally  computationally  private under  initial condition $\xb_0\in\R^{nm}$ with respect to the dataset $\mathpzc{D}$, if and only if there exists $\mathpzc{D}^\prime\neq\mathpzc{D}$ such that $\pdf\big(\mathscr{M}_{\xb_0}(\mathpzc{D})\big) = \pdf\big(\mathscr{M}_{\xb_0}(\mathpzc{D}^\prime)\big)$.
\end{proposition}

In the following, we show some more explicit characterizations of the computational privacy in network linear equation solvers. Let $\Eq=\{\ab^\top\yb=b:\ab\in\R^m,b\in\R\}$ be the space of  linear equations over $\R^m$, and then let $
\Eq^\ast_{[nm]}=\big\{\Ab\yb=\bb:\Ab\in\R^{n\times m},\bb\in\R^n,\ \bb\in\range(\Ab)\big\}$ be the space of the solvable $n$--dimensional linear equations over $\R^m$.

\begin{definition}\label{def:linear_equation_equivalence}
	(i) Let $\eq:\ab^\top\yb=b$ and $\eq^\prime:\ab^{\prime\top}\yb=b^\prime$ be two equations in $\Eq$. We term that $\eq$ and $\eq^\prime$ are equivalent if there exists $\alpha\in\R\setminus\{0\}$ such that $ \ab^\prime = \alpha \ab,\ b^\prime = \alpha b$.
	
	(ii) Let $\eq:\Ab\yb=\bb$ and $\eq^\prime:\Ab^\prime\yb=\bb^\prime$ be two equations in $\Eq^\ast_{[nm]}$. We say that $\eq$ and $\eq^\prime$ are equivalent if there exist $\alpha_1,\dots,\alpha_n\in\R\setminus\{0\}$ such that
	$\Ab^\prime=\diag(\alpha_1,\dots,\alpha_n)\Ab,\ \bb^\prime=\diag(\alpha_1,\dots,\alpha_n)\bb.$
\end{definition}
It is worth mentioning that equipped with the equivalence relation in Definition \ref{def:linear_equation_equivalence}, each equation $\eq\in\Eq^\ast_{[nm]}$ corresponds to a unique equivalence class: $\{\eq^\prime\in\Eq^\ast_{[nm]}:\eq^\prime\sim\eq\}$. This further specifies the quotient space $\Eq^\ast_{[nm]}/\sim:=\big\{\{\eq^\prime\in\Eq^\ast_{[nm]}:\eq^\prime\sim\eq\}:\eq\in\Eq^\ast_{[nm]}\big\}$. It can be easily verified that the linear equations in the same equivalence class share a common solution set. Therefore, with Definition \ref{def:linear_equation_equivalence}, we have effectively identified the information of a solvable linear equation with its solution set for equations in $\Eq^\ast_{[nm]}$. For the two distributed linear equation solvers in (\ref{eq:shi2017}) and (\ref{eq:PCA}), we have the following result.

\begin{theorem}\label{thm1}
	Consider a linear equation $\eq\in\Eq^\ast_{[nm]}$. Assume that the weight matrix $\Wb$ is public knowledge.
	
	(i) Suppose $\alpha$ is also public knowledge. The CPA (\ref{eq:shi2017}) is not    computationally   private in  the quotient space $\Eq^\ast_{[nm]}/\sim$ against global dynamics eavesdroppers for $\eq$    under almost all initial conditions.
	
	(ii) The PCA (\ref{eq:PCA}) is not    computationally   private in   the quotient space $\Eq^\ast_{[nm]}/\sim$ against global dynamics eavesdroppers   for $\eq$ under almost all initial conditions.
\end{theorem}

The following example illustrates how the information about the linear equation has been lost during the recursion of a distributed linear equation solver, which partially explains the intuition behind Theorem~\ref{thm1}.

\noindent{\bf Example 2.} Consider a $4$--node star graph centered at node $1$ and the following linear equation:
\begin{equation}\notag
\eq:\quad\begin{bmatrix}
3 & -1\\
1.5 & 0.8\\
-2 & 1.5\\
-1.2 & 4
\end{bmatrix}\yb=
\begin{bmatrix}
5\\
-0.1\\
-5\\
-9.2
\end{bmatrix}.
\end{equation}
Let the CPA (\ref{eq:shi2017}) with $\alpha=0.1$ and
$$\Wb=\begin{bmatrix}
0.1 & 0.3 & 0.2 & 0.4\\
0.3 & 0.7 & 0 & 0\\
0.2 & 0 & 0.8 & 0\\
0.4 & 0 & 0 & 0.6
\end{bmatrix}$$
starting from some initial condition be implemented.

\begin{figure} [htbp]
	\hspace*{0cm}
	\vspace*{0cm}
	\subfigure[Node trajectories]
	{
		\hskip 0.33cm
		\begin{minipage}{0.45\linewidth}
			\centering
			\includegraphics[width=7.2cm]{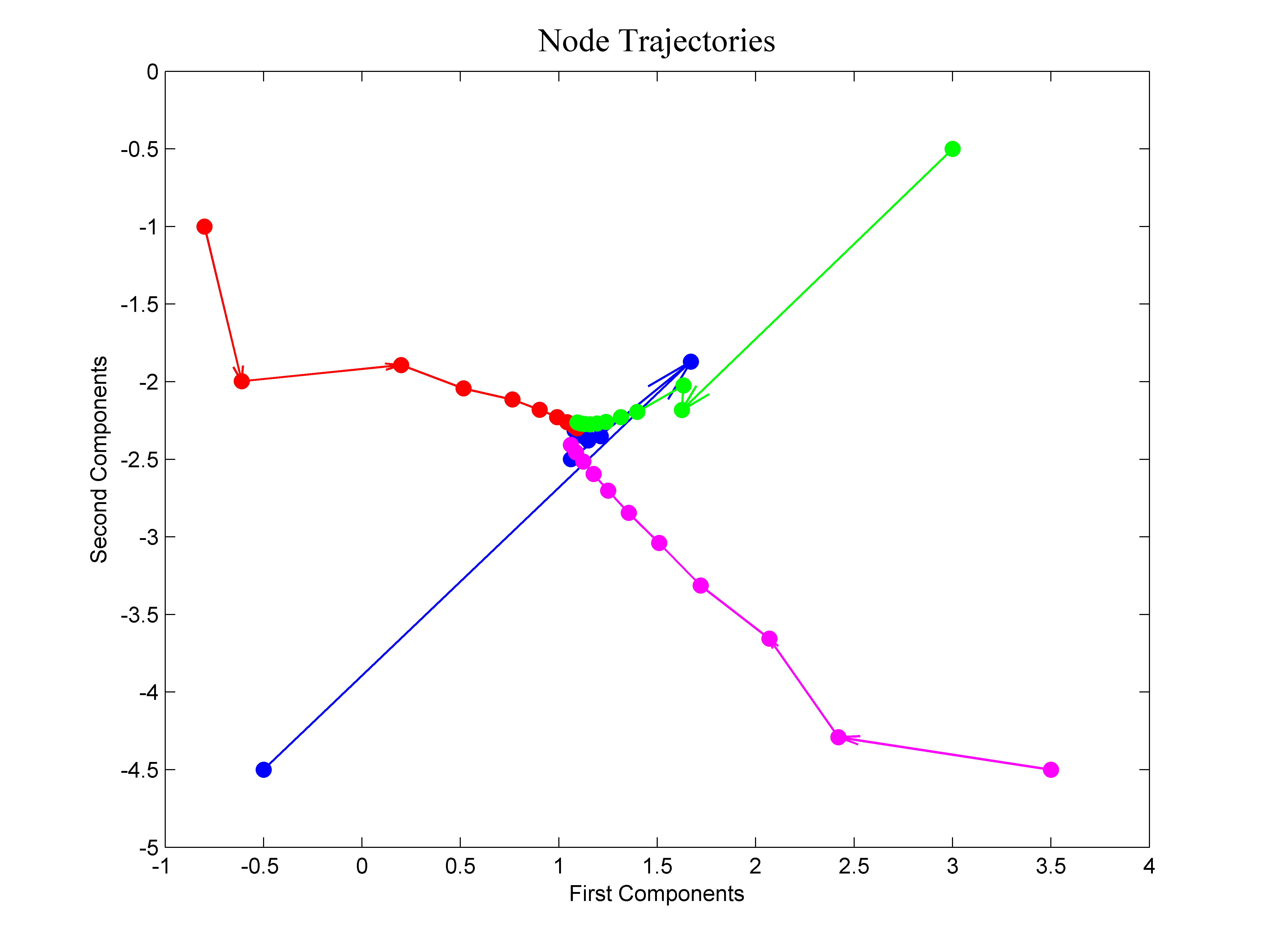}
		\end{minipage}
	}
	\subfigure[Projection computation]
	{
		\hskip 0.33cm
		\begin{minipage}{0.45\linewidth}
			\centering
			\includegraphics[width=7.2cm]{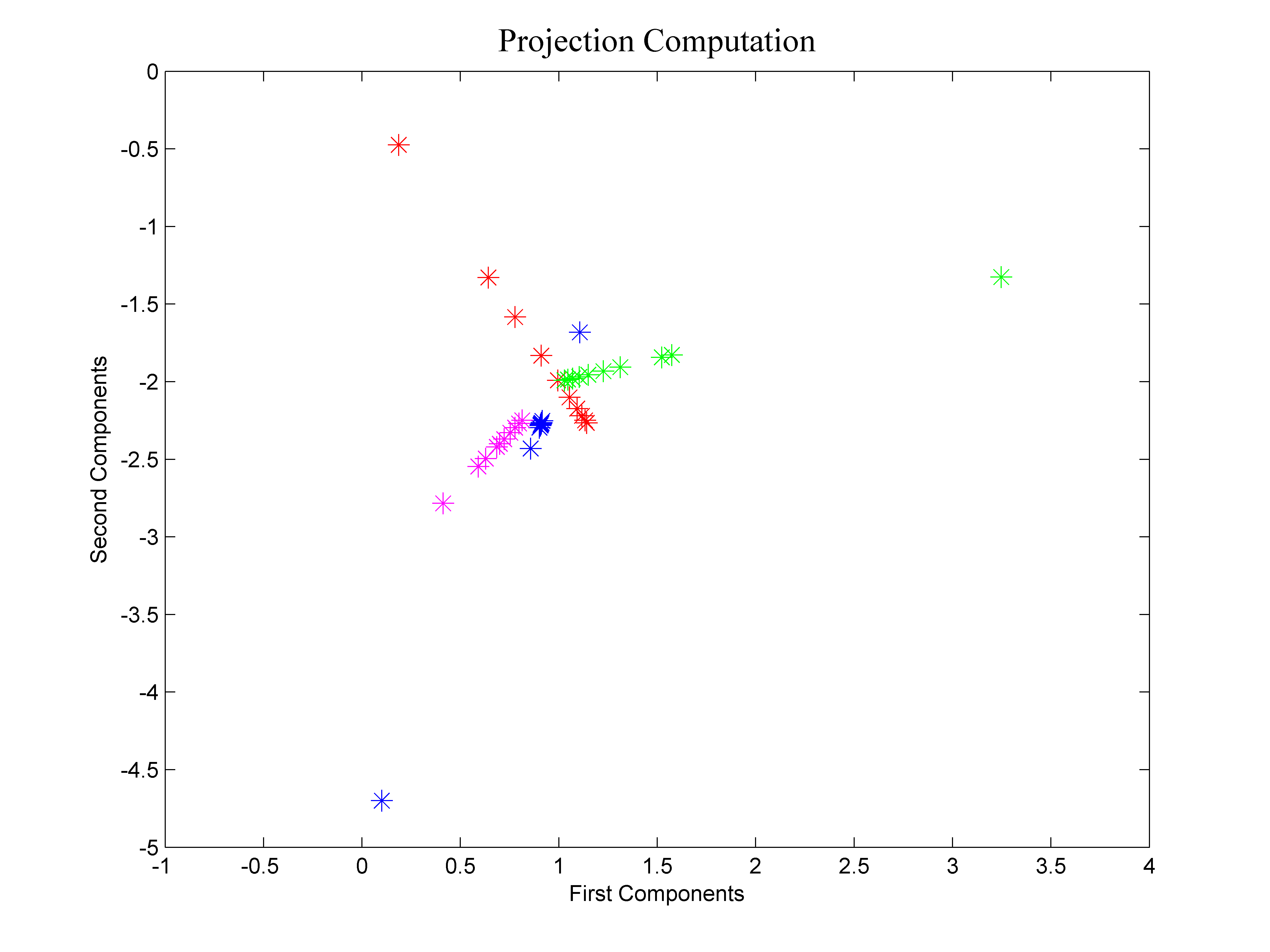}
		\end{minipage}
	}
	\subfigure[Equation reconstruction]
	{
		\hskip 0.32cm
		\begin{minipage}{1\linewidth}
			\centering
			\includegraphics[width=7.2cm]{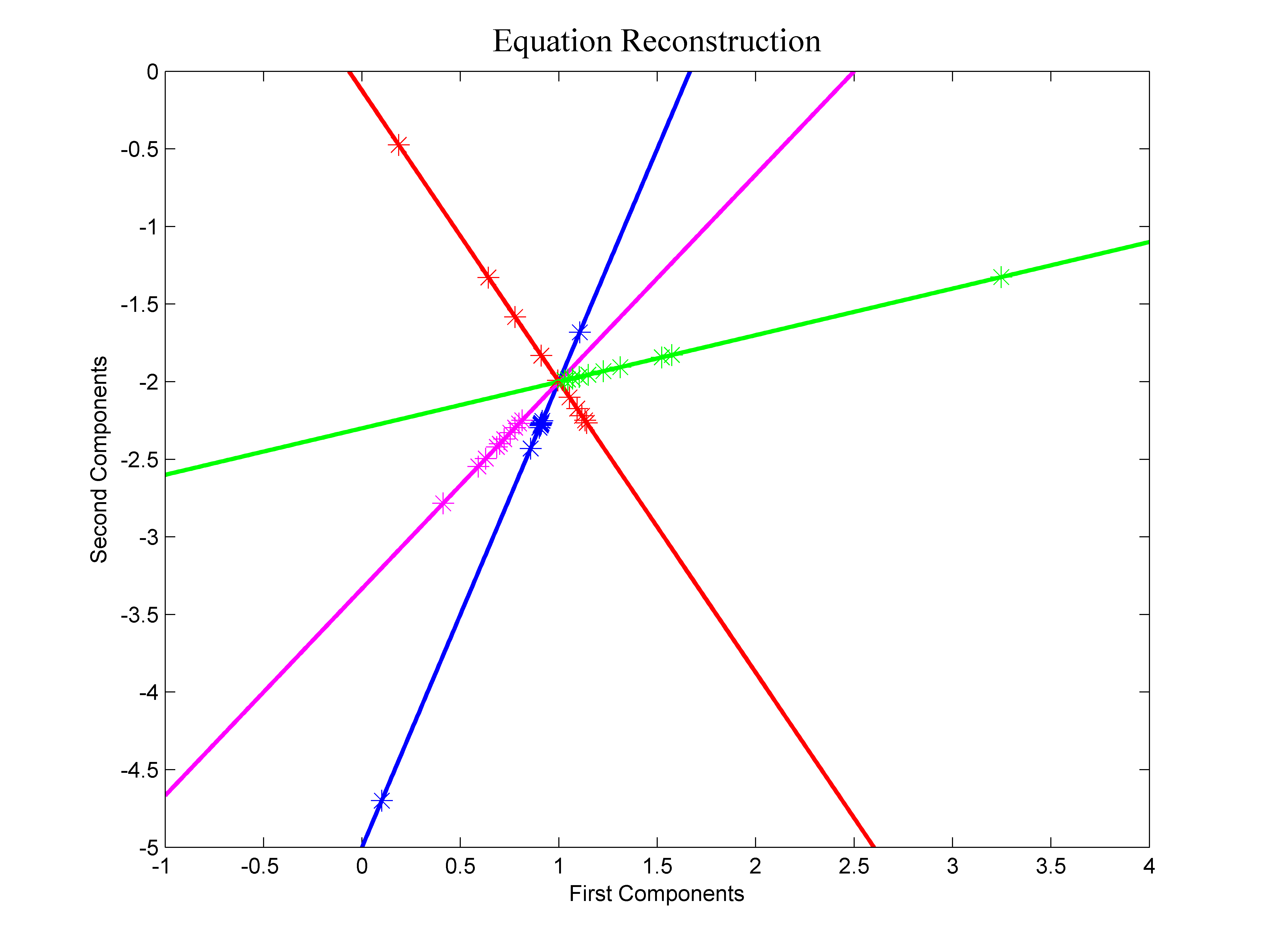}
		\end{minipage}
	}
	\caption{(a) shows the trajectories of the node states $\xb_i(t)$ with $t=0,1,\dots,10$ and $i=1,2,3,4$. (b) demonstrates the computed projection points $\PP_i\big(\xb_i(t)\big),t=0,1,\dots,9$ by eavesdroppers. (c) illustrates the disclosed solution affine space of the linear equations.}
	\label{fig:rec_global}
\end{figure}

A geometric illustration of the way that a global dynamics eavesdropper reconstructs $\eq$ is shown in Figure \ref{fig:rec_global}. In Figure \ref{fig:rec_global} (a), the states $\xb_i(t),t=0,1,\dots,10$ of nodes $i=1,2,3,4$ are plotted by the dot points of different colors. Then the projections $\PP_i\big(\xb_i(t)\big),t=0,1,\dots,9$ for all $i$ are computed by the eavesdropper according to (\ref{eq:shi2017}), as marked by the stars in Figure \ref{fig:rec_global} (b). By $\PP_i\big(\xb_i(t)\big)$ and the perpendicularity of the projections, the solution spaces of $\eq_i^\prime$s are identified as shown in Figure \ref{fig:rec_global} (c). As a result, an equivalent set $\{\eq^\prime:\eq^\prime\sim\eq\}$ has been reconstructed by the eavesdropper.

\subsection{Differentiallly Private Computation}

From the above investigation of distributed linear equation solvers, we now have a concrete example where the local datasets for a distributed computation setup may indeed be lost during the computation process. We acknowledge that  information nodes have to reveal certain amount of information about the local datasets $\mathpzc{D}_i$ in the overall closed-loop dynamical equation (\ref{eq:dcp}): If absolutely no information about $\mathpzc{D}_i$ enters (\ref{eq:dcp}), the eventual computation outcome will not depend on $\mathpzc{D}_i$ and therefore will never be a network-wide solution. It has been shown  how we can borrow the idea of differential privacy in distributed optimization problems, e.g.,  \cite{huang2015differentially,pappas2017}, following the earlier work on differentially private filtering \cite{pappas2014}.

\subsubsection{Differentiallly Private Distributed Computation}

Differential privacy is a concept that statistically quantifies the ability of an information protocol to protect the privacy of confidential databases, and meanwhile provide as accurate query output as possible. It was formally proposed in \cite{dwork2006calibrating}, with the intuition being that a query that comes from a database is nearly indistinguishable even if a record in the database is modified or removed. Recall that the dataset  $\mathpzc{D}$ is allocated over the network by partition $\mathpzc{D}= \mathpzc{D}_1 \otimes \cdots \otimes \mathpzc{D}_n$. We say $\mathpzc{D}'= \mathpzc{D}_1' \otimes \cdots \otimes \mathpzc{D}_n'$ is adjacent with $\mathpzc{D}$  if there exists $i\in\{1,\dots,n\}$ such that
$$ \mathpzc{D}_j= \mathpzc{D}_j^\prime,\quad \forall j\neq i.$$
Recall that $\mathscr{M}_{\xb_0}(\cdot)$ is the mechanism along a randomized distributed computing protocol defined by $
\mathscr{M}_{\xb_0}(\mathpzc{D})=\big(\xb(t)\big)_{t=0}^\infty.
$ We can then borrow conventional differential privacy notion to define that $\mathscr{M}_{\xb_0}$ is  $\epsilon$-differentially private with some privacy budget $\epsilon>0$ if for all adjacent $\mathpzc{D}, \mathpzc{D}^\prime$, and for all $\mR\subset\range(\mathscr{M}_{\xb_0})$,
\begin{equation}\notag
\Pr(\mathscr{M}_{\xb_0}(\mathpzc{D})\in\mR) \le e^{\epsilon} \cdot \Pr(\mathscr{M}_{\xb_0}(\mathpzc{D}^\prime)\in\mR),
\end{equation}
where the probability is taken over the randomness used by $\mathscr{M}_{\xb_0}$.

In the differential privacy literature \cite{dwork2006calibrating}, one popular approach is the so--called Laplace mechanism, where Laplace noise with appropriate randomness is added to the output returned from the input data. One can apply this approach to the distributed computing protocol  (\ref{eq:dcp}) by simply viewing $\mathscr{M}_{\xb_0}$ as a data processing mechanism and resolving differential privacy challenges with added noise \cite{huang2015differentially,pappas2017}.  However, adding noise to (\ref{eq:dcp}) will in general  make it impossible for the protocol to converge to the true value, and in many cases one can only establish error guarantee that is increasingly conservative as time moves forward under rather restrictive conditions \cite{huang2015differentially,pappas2017}. In contrast, classical differential privacy results manage to achieve differential privacy with tolerable errors that vanish with an inceasing dimension of the input data \cite{dwork2006calibrating}. The reason for this  sharp comparison is that, the  dimension of the outputs of $\mathscr{M}_{\xb_0}$ (i.e., the $\mathbf{x}(t)$) increases linearly with time steps, controlling the error   the way how a standard differentially private query system works becomes challenging or even impossible. In the following, we illustrate this point for distributed LAE solvers.

\subsubsection{Differentially Private Network LAE Solving}
For network linear equations, we can define adjacency properties of different datasets over the quotient space.

\begin{definition}
	Consider two linear equations $\eq:\Ab\yb=\bb,\ \eq^\prime:\Ab^\prime\yb=\bb^\prime\ \in\Eq^\ast_{[nm]}$. We call them to be $(\dA,\db)$--adjacent if there exists $i\in\{1,\dots,n\}$ such that
	$ \mathcal{E}_j\sim\mathcal{E}_j^\prime,\forall j\neq i$ and
	$$\bigg\|\frac{\Ab_i\Ab_i^\top}{\Ab_i^\top\Ab_i}-\frac{\Ab_i^\prime\Ab_i^{\prime\top}}{\Ab_i^{\prime\top}\Ab_i^\prime}\bigg\|\le\dA,\ \bigg\|\frac{\bb_i\Ab_i}{\Ab_i^\top\Ab_i}-\frac{\bb_i^\prime\Ab_i^\prime}{\Ab_i^{\prime\top}\Ab_i^\prime}\bigg\|\le\db.$$
\end{definition}

\begin{definition}Consider a distributed solver for linear equations in the space $\Eq^\ast_{[nm]}$.
	Let $\mathscr{M}_{\xb_0}(\cdot)$ be the resulting mechanism   with $
	\mathscr{M}_{\xb_0}(\mathcal{E})=\big(\xb(t)\big)_{t=0}^\infty.
	$ Then the solver is called $\epsilon-$differenially private under $(\dA,\db)$--adjacency if
	\begin{equation}\notag
	\Pr(\mathscr{M}_{\xb_0}(\mathcal{E})\in\mR) \le e^{\epsilon} \cdot \Pr(\mathscr{M}_{\xb_0}(\mathcal{E}^\prime)\in\mR),
	\end{equation}
	for all  $(\dA,\db)$--adjacent linear equations $\eq:\Ab\yb=\bb,\ \eq^\prime:\Ab^\prime\yb=\bb^\prime\ \in\Eq^\ast_{[nm]}$, for any initial condition $\xb(0)\in\R^{nm}$, and for all network state trajectories in $\range(\mathscr{M}_{\xb_0})$.
\end{definition}

We assume the network of nodes have prior knowledge about a linear equation $\eq$ that one of its solutions falls in a compact and convex set $\Omega\subset\R^m$. We define a projection $\PPo:\R^m\to\R^m$ with $\PPo(\ub)=\inf\limits_{\vb\in\Omega}\|\vb-\ub\|$. By adapting the deterministic CPA (\ref{eq:shi2017}), we now present the following algorithm as a Laplace--mechanism--based differentially private distributed LAE solver, where Laplace noise $\Lap^m(c\phi^t)$  with $c>0,\ 0<\phi<1$ is injected prior to node state broadcasting at each time $t$. Define $\alpha(t)=\lambda\psi^t$ with $\lambda>0$ and $0<\psi<1$  as a time--varying step size. Following the idea of  \cite{huang2015differentially}, we propose the following differentially private distributed linear equation solver based on the Laplace approach.
\begin{algorithm}[ht]{{\bf DP--DLES} Differentially Private Distributed Linear Equation Solver}
	\begin{algorithmic}[1]
		\STATE Set $t\gets0$ and $\xb(0)$.
		\STATE Each node $i\in\mV$ computes $\xbf_i(t)=\PPo\big(\xb_i(t)\big)$.
		\STATE Each node $i\in\mV$ draws $\omegab_i(t)$ from the distribution $\Lapm(c\phi^t)$.
		\STATE Each node $i\in\mV$ computes $\xbs_i(t)\gets\xbf_i(t)+\omegab_i(t)$ and propagates $\xbs_i(t)$ to all its neighbors $j\in\mN_i$.
		\STATE Each node $i\in\mV$ computes $\xb_i(t+1) \gets \sum\limits_{j\in\mN_i}\Wb_{ij}\xbs_j(t)+\alpha(t)\bigg(\PP_i\big(\xbf_i(t)\big)-\xbf_i(t)\bigg)$.
		\STATE Set $t\gets t+1$ and go to Step 2.
	\end{algorithmic}
\end{algorithm}

We present the following results.
\begin{theorem}\label{thm:sensitivity_dp_CPA}
	Suppose $\Wb\in\R^{n\times n}$ has full rank. Let $B=\sup\limits_{\vb\in\Omega}\|\vb\|$. Then the DP--DLES preserves $\epsilon-$differenial privacy  under $(\dA,\db)$--adjacency if $\psi<\phi$ with
	\begin{equation}\label{eq:dp}
	\frac{\phi}{\phi-\psi}\cdot  \frac{\lambda}{c} \cdot \frac{  \sqrt{nm}(B\dH+\dz)}{\sigm(\Wb)} \leq \epsilon.
	\end{equation}
\end{theorem}
Note that $\epsilon$ is the privacy budget which is expected to be a small number. As a result, for a fixed amount of noise injected to the computation process represented by $c$ and $\phi$ in $\Lapm(c\phi^t)$, $\lambda$ and $\psi$ to be small enough for (\ref{eq:dp}) to hold. The values of $\lambda$ and $\psi$, however, determine how much local data information $\mathpzc{D}_i$ is used in the computation step through the $\mathpzc{P}_i$. Therefore, small $\lambda$ and $\psi$ means greater computation error will take place. This reveals the privacy vs computation dilemma, which can be seen from the following example.

\noindent {\bf Example 3.} Consider the same linear equation, network structure and weight matrix $\Wb$ as Example 1. We let $\Omega$ be a ball centered at $\yb^\ast$ with radius one. We tune the parameters of DP--DLES such that DP--DLES preserves $(\epsilon,1,1)$--differential privacy with $\epsilon=2,4,6,8$. Corresponding to the parameter setup for  each $\epsilon$, we independently execute DP--DLES from starting from initial states sampled from $[-1,1]$ at random for $10^3$ times, and plot the averaged $$ \|\sum_{i=1}^4\xb_i(t)-\yb^\ast\|
$$ in Figure \ref{fig:dp}. It is clear from the numerical result  that when better differential privacy is achieved,
worse computational accruracy is recorded.

\begin{figure}[ht]
	\centering
	\includegraphics[width=9cm]{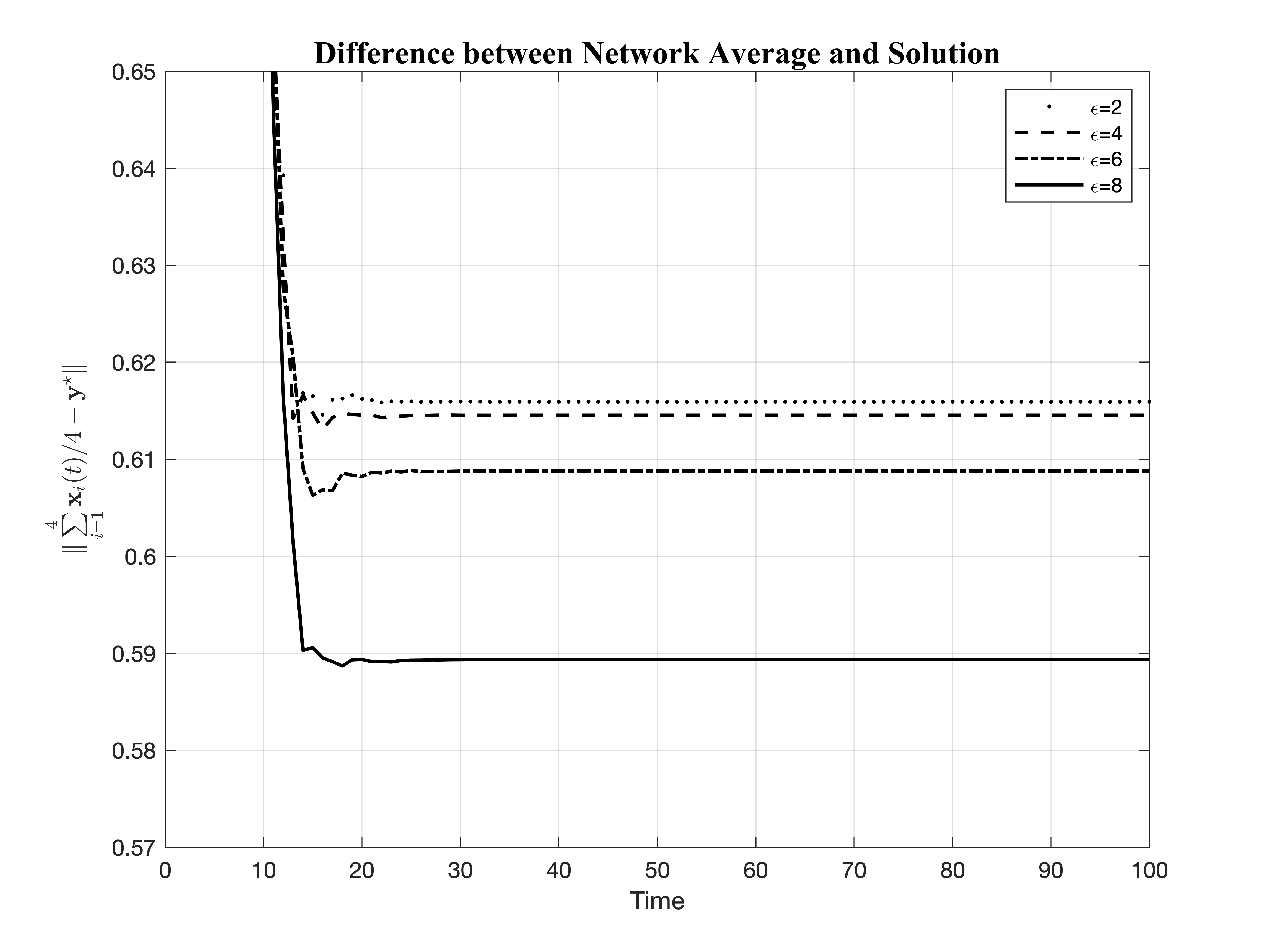}
	\caption{The plot of the averaged $\|\sum_{i=1}^4\xb_i(t)-\yb^\ast\|$ along DP--DLES with $\epsilon=2,4,6,8$.}
	\label{fig:dp}
\end{figure}

\section{PPSC Mechanism}\label{Sec:PPSC}

In view of the above discussions, we now see that on one hand, dynamical privacy indeed exists in distributed computing processes; and on the other hand, adding noise to the computational recursions improves privacy but in the meantime jeopardizes the computation accuracy. Noticing the fact that in average consensus and in consensus-based distributed computing algorithms, the consensus manifold is where the network-wide solution lies in, for which the node states can stay by local interactions even
in the presence of additional random noise. In this section, we show that  we can effectively use this observation to hide node privacy, under a so-called PPSC mechanism.

\subsection{Mechanism Definition}
\begin{definition}
	Let $\beta_i\in\mathbb{R}^m$ be held by node $i$ over the network $\mathrm{G}=(\mathrm{V},\mathrm{E})$.
	An algorithm running over the network is called a distributed  Privacy-Preserving-Summation-Consistent (PPSC) mechanism, which produces output $\bm{\beta}^\sharp=(\beta_1^{\sharp\top}\ \dots\ \beta_n^{\sharp\top})^\top$ from the network input $\bm{\beta}=(\beta_1^\top\ \dots\ \beta_n^\top)^\top$, if the following conditions hold:
	\begin{itemize}
		\item [(i)] {\em (Graph Compliance)} Each node $i$ communicates only  with its neighbors in the set $\mathrm{N}_i$;
		\item[(ii)] {\em (Local Privacy Preservation)} Each node $i$ never reveals its initial value $\beta_i$ to any other agents or a third party;
		\item[(iii)] {\em (Global Privacy Preservation)} $\bm{\beta}$ is non-identifiable given $\bm{\beta}^\sharp$ with even an infinite number of independent samples;
		\item[(iv)] {\em (Summation Consistency)} $\sum\limits_{i=1}^n \beta^{\sharp}_i=\sum\limits_{i=1}^n \beta_i$.
	\end{itemize}
\end{definition}

The condition (i) requires that the information flow of the algorithm must comply with the underlying graph; the condition (ii) says that during running of the algorithm no node will have to directly reveal its initial value $\beta_i$; the condition (iii) further suggests that even if final node states are known, it should not be possible to use that knowledge to recover $\beta_i$; the condition (iv) asks for that sum of the output of the algorithm should be consistent with that of the input. The overall mechanism of PPSC is denoted as the following mapping:
\begin{align}
\bm{\beta}^\sharp=\mathscr{P}(\bm{\beta}).
\end{align}
\subsection{PPSC for Distributed Computation}\label{sec:app}
Now we show that the PPSC mechanism can be used as a privacy preservation step for distributed computing algorithms.

\medskip

\noindent{\bf [Average Consensus]} For the average consensus algorithm (\ref{eq:ave}), we can   run a distributed PPSC algorithm over $(\beta_1,\dots,\beta_n)$ and generate $(\beta_1^\sharp,\dots,\beta_n^\sharp)$. Letting $\mathbf{x}_i(0)=\beta_i^\sharp,\ i\in\mV$ in instead, there still holds $\lim\limits_{t\to\infty}\mathbf{x}_i(t)=\sum\limits_{j=1}^n\beta_j/n$, but $\beta_i$ has been kept as private information for node $i$.

\medskip

\noindent{\bf [Network Linear Equations]}
For the network linear equation (\ref{eq:LAE}), the PPSC mechanism can provide    distributed solver with privacy guarantees. \medskip
\begin{algorithm}[H]
	{$\mathbf{Privacy\textnormal{ }Preserving\textnormal{ }Linear\textnormal{ }Equation\textnormal{ }Solver}$}
	\begin{algorithmic}[1]
		\STATE Set $t\gets 0$ and $y_1(t),\dots,y_n(t)\in\R^r$;
		\STATE Run a PPSC algorithm with input $y_1(t),\dots,y_n(t)$ and produce output $y^{\sharp}_1(t),\dots,y^{\sharp}_n(t)$;
		\STATE Run the averaging consensus algorithm with input $y^{\sharp}_1(t),\dots,y^{\sharp}_n(t)$ and output the agreement $\bar{y}(t)$ at each node $i$;
		\STATE Each node $i$ computes $y_i(t+1)\gets\mathpzc{P}_i(\bar{y}(t))$. Set $t\gets t+1$ and go to Step 2.
	\end{algorithmic}
\end{algorithm}
The above privacy preserving linear equation solver produces the following recursion:
\begin{equation}\label{eq:pp_linear_equation_solver}
y_i(t+1)=\mathpzc{P}_i\bigg(\sum\limits_{j=1}^n \frac{y_j(t)}{n}\bigg),\ t\in\Z^{\ge 0},\ i\in\mV,
\end{equation}
which falls to the category of the projected consensus algorithms in \cite{nedic10}. According to Proposition 2 in \cite{nedic10}, each node state $y_i(t)$ converges to a solution of the linear equation (\ref{eq:LAE}) if it admits exact solutions, i.e., $\zb\in\spa\{\Hb\}$.

\medskip

\noindent{\bf [Distributed Gradient Descent]}
With the PPSC mechanism, the distributed gradient descent algorithm (\ref{eq:gradient}) can also be made privacy-preserving, by the following procedure.

\begin{algorithm}[H]
	{$\mathbf{Privacy\textnormal{ }Preserving\textnormal{ }Distributed\textnormal{ }Optimizer}$}
	\begin{algorithmic}[1]
		\STATE Set $t\gets 0$ and $y_1(t),\dots,y_n(t)\in\R^r$;
		\STATE Each node $i$ computes $y_i^\flat(t)\gets y_i(t)-\frac{1}{\sqrt{t+1}}\nabla f_i(y_i(t))$;
		\STATE Run a PPSC algorithm with input $y^{\flat}_1(t),\dots,y^{\flat}_n(t)$ and produce output $y^{\sharp}_1(t),\dots,y^{\sharp}_n(t)$;
		\STATE Run the averaging consensus algorithm with input $y^{\sharp}_1(t),\dots,y^{\sharp}_n(t)$ and output the agreement $\bar{y}(t)$ at each node $i$;
		\STATE Each node $i$ sets $y_i(t+1)\gets\bar{y}(t)$.
		\STATE Set $t\gets t+1$ and go to Step 2.
	\end{algorithmic}
\end{algorithm}
\noindent The underlying dynamics of the privacy-preserving distributed optimizer is described by
\begin{equation}\label{eq:pp_distributed_opt}
y_i(t+1)=\sum\limits_{j=1}^n \frac{y_j(t)}{n}-\frac{1}{\sqrt{t+1}}\sum\limits_{j=1}^n\frac{\nabla f_j(y_j(t))}{n}.
\end{equation}

\subsection{Privacy Preservation by PPSC}

From the above three examples of integrating a PPSC step in each round of recursion of distributed computing algorithms,  we can see that the PPSC mechanism effectively achives the following advantages in a distributed computing protocol:
\begin{itemize}
	\item[(i)] Nodes keep their real states $\mathbf{x}_i(t)$ strictly from themselves, and receive only virtual states  $\mathbf{x}_i^\sharp(t)$  produced by the PPSC mechanism $\mathscr{P}$. The impossibility of identifying $\mathbf{x}(t)$ from $\mathbf{x}^\sharp(t)=\mathscr{P}(\mathbf{x}(t))$ (with even an arbitrary number of realizations of $\mathbf{x}^\sharp(t)$ for each $t$) provides the nodes with plausible deniability about their real states $\mathbf{x}(t)$ for all $t$, which immediately implies the same deniability about the $\mathpzc{D}_i$.
	\item[(ii)] The price paid in the PPSC step is that averaging the $\mathbf{x}^\sharp(t)=\mathscr{P}(\mathbf{x}(t))$ may take more time than averaging the $\mathbf{x}(t)$ over the same graph $\mathrm{G}$.
	Since $\mathbf{x}_i^\sharp(t)$ has proven privacy guarantee, there is no need to share it with only neighbors for privacy concerns. Each node $i$ can broadcast $\mathbf{x}_i^\sharp(t)$ to the entire network whenever possible, and therefore may compensate the efficiency lost in the averaging step.
	\item[(iii)] The sum consistency property from PPSC mechanism guarantees that the computation produces the same result  with or without the PPSC privacy preservation steps in the above consensus-based linear equation solving and gradient descent algorithms.
\end{itemize}
In Part II of the paper, we will show that conventional gossip algorithms can be used to realize PPSC mechanism in distributed manners over a network by only  a finite number of steps, and beyond non-identifiability of the PPSC mechanism $\mathscr{P}$, we can even easily ensure   $\mathscr{P}$ to be differentially private.

We also acknowledge that a weakness of the above PPSC-based distributed computation  is that it involves running PPSC and average consensus to completion for every gradient step, which is likely to be very slow over the network $\mathrm{G}$ and lead to truncation errors when we use only a finite number of average consensus. However, since $\betab^\sharp$ can be proven  privacy preserving, each node $i$ can broadcast $\beta_i$ to the entire network for the average consensus, instead of communicating with neighbors in $\mathrm{G}$ only. As a result, averaging $\betab^\sharp$ can be   much faster compared to averaging $\betab$ over $\mathrm{G}$, producing  the same consensus value. This would not resolve the effect of finite-step  truncation errors, which is left to future work.

\section{Local Dynamical Privacy}\label{sec:local}
Finally, we turn our attention to local dynamics eavesdroppers, where two types of behaviors, passive or active, are distinguished for the role of a local eavesdropper in the network dynamics.

\subsection{Local Dynamics Eavesdroppers}

We introduce the following definition.

\begin{definition}\label{def:local_eavesdroppers}
	(i) A passive local dynamics eavesdroppers is a party adherent to one node $i\in\mV$ and therefore has access to the observations of $\{\xb_j(t),j\in\mN_i\}_{t=0}^\infty$.
	
	(ii) An active local dynamics eavesdroppers is a party adherent to one node $i\in\mV$, has access to the observations of $\{\xb_j(t),j\in\mN_i\}_{t=0}^\infty$,  and can arbitrarily  alter the dynamics of $\xb_i(t)$.
\end{definition}
In either of the two cases in Definition \ref{def:local_eavesdroppers}, we assume without loss of generality that such a local dynamics eavesdropper is precisely a node $i\in\mV$. From the perspectives of local eavesdroppers, recovering information about the $\mathpzc{D}$ from $\{\xb_j(t),j\in\mN_i\}_{t=0}^\infty$ defines a structured system identification problem.
Let $\mN_i=\{j_1,\dots,j_{\left|\mN_i\right|}\}$ with $j_1<\dots<j_{\left|\mN_i\right|}$. Define $\Eb_i\in\R^{\left|\mN_i\right|\times n}$ as the matrix whose entries in the $k$--th row are all zeros except for the $k j_k$--th entry being one. Denote $\Fb=\Wb\otimes\Ib_m-\alpha\diag\bigg(\frac{\Hb_1\Hb_1^\top}{\Hb_1^\top\Hb_1},\dots,\frac{\Hb_n\Hb_n^\top}{\Hb_n^\top\Hb_n}\bigg)$. 	Introduce 
\begin{align*}
\mathcal{T}=\big\{\Tb\in\R^{nm\times nm}:\Tb\textnormal{ is invertible, } (\Eb_i\otimes\Ib_m)\Tb=[\Ib_{m\left|\mN_i\right|}\ 0]\big\}.
\end{align*} For the CPA (\ref{eq:shi2017}) against passive local dynamics eavesdroppers, we can establish the following understanding based on the result on obtaining the minimal realization of LTI systems from input--output data in time domain established in \cite{budin1971minimal}.

\begin{theorem}\label{thm2}
	Consider the distributed linear equation solver CPA (\ref{eq:shi2017}) for $\eq\in\Eq^\ast_{[nm]}$. Assume that the weight matrix $\Wb$ and the step size $\alpha$ are public knowledge. Suppose the following two conditions hold: (i) $(\Fb,\Eb_i\otimes\Ib_m)$ is a completely observable pair; (ii) There exist $t_1<\dots<t_{nm}$ such that the vectors $\xb(t_k),\ k=1,\dots,nm$ are linearly independent.
	Then for the passive local dynamics eavesdropper $i$, it can be determine from $\{\xb_j(t),j\in\mN_i\}_{t=0}^\infty$ a subset $\Eq^\dagger\subset\Eq^\ast_{[nm]}$ such that $\eq\in\Eq^\dagger$, where
	\begin{align*}\notag
	\Eq^\dagger=\bigg\{\Ab\yb=\bb:\diag\bigg(\frac{\Ab_1\Ab_1^\top}{\Ab_1^\top\Ab_1},\dots,\frac{\Ab_n\Ab_n^\top}{\Ab_n^\top\Ab_n}\bigg)=(\Wb\otimes\Ib-\Tb^{-1}\Fb_\ast \Tb)/\alpha;\bb=\Ab\lim\limits_{t\to\infty}\xb_i(t);\ \Tb\in\mathcal{T}\bigg\}.
	\end{align*}
	Here $\Fb_\ast \in\R^{nm\times nm}$ depends on $\{\xb_j(t),j\in\mN_i\}_{t=0}^\infty$ only.
\end{theorem}

For the CPA (\ref{eq:shi2017}) against an active local eavesdropper $i$, we provide the following result based on the time--domain--based system identification method proposed in \cite{mckelvey1995subspace}.
\begin{theorem}\label{thm3}
	Consider the distributed linear equation solver CPA (\ref{eq:shi2017}) for $\eq\in\Eq^\ast_{[nm]}$. Assume that $\eq$ admits a unique solution $\yb^\ast\in\R^m$. Assume further that the weight matrix $\Wb$, the step size $\alpha$  and the solution $\yb^\ast$ are public knowledge. Then there is a simple strategy for an active local dynamics eavesdropper $i$ by adding a periodic signal $\rb:\Z^{\ge0}\to\R^m$ to $\xb_i(t)$ at each time $t$ under which it can identify $\Eq^\dagger\subset\Eq^\ast_{[nm]}$ with $\eq\in\Eq^\dagger$, given by
	\begin{align*}\notag
	\Eq^\dagger&=\bigg\{\Ab\yb=\bb:\diag\bigg(\frac{\Ab_1\Ab_1^\top}{\Ab_1^\top\Ab_1},\dots,\frac{\Ab_n\Ab_n^\top}{\Ab_n^\top\Ab_n}\bigg)=(\Wb\otimes\Ib-\Tb^{-1}\Fb_\ast \Tb)/\alpha;\bb=\Ab\yb^\ast;\ \Tb\in\mathcal{T}\bigg\},\\
	\mathcal{T}&=\big\{\Tb\in\R^{nm\times nm}:\Tb\textnormal{ is invertible, } (\Eb_i\otimes\Ib_m)\Tb=\Cb_\ast \big\}.
	\end{align*}
	Here $\Fb_\ast \in\R^{nm\times nm}$ and $\Cb_\ast \in\R^{\left|\mN_i\right|m\times nm}$ depend on $\rb(t)$ and $\{\xb_j(t),j\in\mN_i\}_{t=0}^\infty$ only.

\end{theorem}

From Theorem \ref{thm2} and Theorem \ref{thm3}, a passiave or active eavesdropper eventually finds $\Ab_\ast\in\R^{nm\times nm}$ and $\Cb_\ast\in\R^{\left|\mN_i\right|m\times nm}$ such that $\tH=\diag\bigg(\frac{\Hb_1\Hb_1^\top}{\Hb_1^\top\Hb_1},\dots,\frac{\Hb_n\Hb_n^\top}{\Hb_n^\top\Hb_n}\bigg)$ satisfies
\begin{align}
(\Wb\otimes\Ib_m-\alpha\tH )\Tb &= \Tb\Ab_\ast,\notag\\
(\Eb_i\otimes\Ib_m)\Tb&=\Cb_\ast .\notag
\end{align}
Introduce a new matrix variable $\Qb = \tH\Tb$. Then solving $\Qb$ and $\Tb$ from the above equations will lead to the recovering of $\tH$, which further yilds $\mathcal{E}$ in the quotient space in view of the knowlege of $\mathbf{y}^\ast$.  Then 
 we can   establish by vectorization that 
\begin{align}\label{r100}
\begin{bmatrix}
   \mathbf{S}_\ast & -\alpha \Ib_{n^2m^2} \\[6pt]
  \Ib_{nm}\otimes(\Eb_i\otimes\Ib_m)     & 0 
\end{bmatrix}\begin{bmatrix}
    {\rm vec}(\Tb)
    \\[6pt]
   {\rm vec}(\Qb)
\end{bmatrix}= \begin{bmatrix}
    0
    \\[6pt]
 {\rm vec}(\Cb_\ast)
\end{bmatrix}
\end{align}
where $\mathbf{S}_\ast=\Ib_{nm}\otimes(\Wb\otimes\Ib_m)-\Ab_\ast\otimes \Ib_{nm}$. This linear equation about $$
({\rm vec}(\Tb)^\top,{\rm vec}(\Qb)^\top)^\top
$$ is underdetermined, and therefore  exploring the linear structure of the  obtained equations is not enough for the eavesdroppers to recover   $\mathcal{E}$.  However, the  $\Qb$ and $\Tb$ have additional constraints:  $\Tb$ must be non-singular, and $\Qb \Tb^{-1}$ has rank-one diagonal blocks. These highly nonlinear but strong constraint condtions might eventually locate  a unique solution from the solution set of (\ref{r100}).

\subsection{Numerical Example}

\noindent{\bf Example 4.} Let the CPA (\ref{eq:shi2017}) be implemented with the same setup in Example 1 on the graph in Figure \ref{fig:rec_global} for solving a linear equation $\eq:\Hb\yb=\zb$ with
\begin{equation}\notag
\Hb=\begin{bmatrix}
71.5 & -65.5\\
-95 & 47.1\\
-35.5 & 100\\
86.5 & -69
\end{bmatrix},\
\zb=\begin{bmatrix}
-202.5\\
189.2\\
235.5\\
-224.5
\end{bmatrix}.
\end{equation}
Suppose a local eavesdropper is at node $2$ and system identification methods in \cite{budin1971minimal,mckelvey1995subspace} have led to
{\begin{align*}
\Ab_\ast &=
\begin{bmatrix}
0.86 &	1.09 &	-0.87 &	-0.73 &	0.47 &	0.05 & 0.61 & -0.87\\
0.61 &	0.59 &	-0.47 &	-0.16 &	-0.36 &	-0.70 & -0.61 & 0.20\\
0.78 &	1.10 &	-0.94 &	-1.06 &	0.06 &	-0.65 &	0.15 & -0.75\\
1.03 &	0.64 &	-1.12 &	0.27 &	-0.28 &	-0.85 &	-0.68 &	0.09\\
-0.73 &	-1.37 &	1.72 &	1.18 &	0.30 &	0.38 &	-0.53 &	1.19\\
-1.22 &	-0.78 &	1.40 &	0.84 &	0.47 &	2.10 &	0.92 &	-0.01\\
2.03 &	1.60 &	-2.97 &	-1.82 &	-0.33 &	-1.94 &	0.03 &	-0.96\\
0.36 &	0.19 &	-0.35 &	-0.21 &	-0.12 &	-0.39 &	-0.30 &	0.80
\end{bmatrix},\\
\Cb_\ast  &=
\begin{bmatrix}
-60.68 &	83.44 &	6.16 &	-67.56 &	37.84 &	7.67 &	63.46 &	-63.63\\
-49.78 &	-42.83 &	55.83 &	58.86 &	49.63 &	99.23 &	73.74 &	-47.24\\
23.21 &	51.44 &	86.80 &	-37.76 &	-9.89 &	-84.36 &	-83.11 &	-70.89\\
-5.34 &	50.75 &	-74.02 &	5.71 &	-83.24 &	-11.46 &	-20.04 &	-72.79
\end{bmatrix}.
\end{align*}}

Let the eavesdropper  attempt to approximate  $\tH $ by the following optimization problem:
\begin{equation} \label{eq:local_optimization}
\begin{aligned}
\min_{\Hb\in\R^{n\times m},\Tb\in\R^{nm\times nm}}\qquad &  \|(\Wb\otimes\Ib_m-\alpha\tH )\Tb-\Tb\Ab_\ast\|^2_{\rm F}\\
{\rm s.t.} \qquad  &   (\Eb_i\otimes\Ib_m)\Tb=\Cb_\ast \\
& \tH = \diag\bigg(\frac{\Hb_1\Hb_1^\top}{\Hb_1^\top\Hb_1},\dots,\frac{\Hb_n\Hb_n^\top}{\Hb_n^\top\Hb_n}\bigg).
\end{aligned}
\end{equation}
It turns out that the interior-points method leads to iterations converging to local minima in general. However, when the initial estimate of the interior--point method is sufficiently close to the true $\Hb$, the interior--point method can find the value of $\Hb$.

\begin{figure}[H]
	\centering
	\includegraphics[width=9cm]{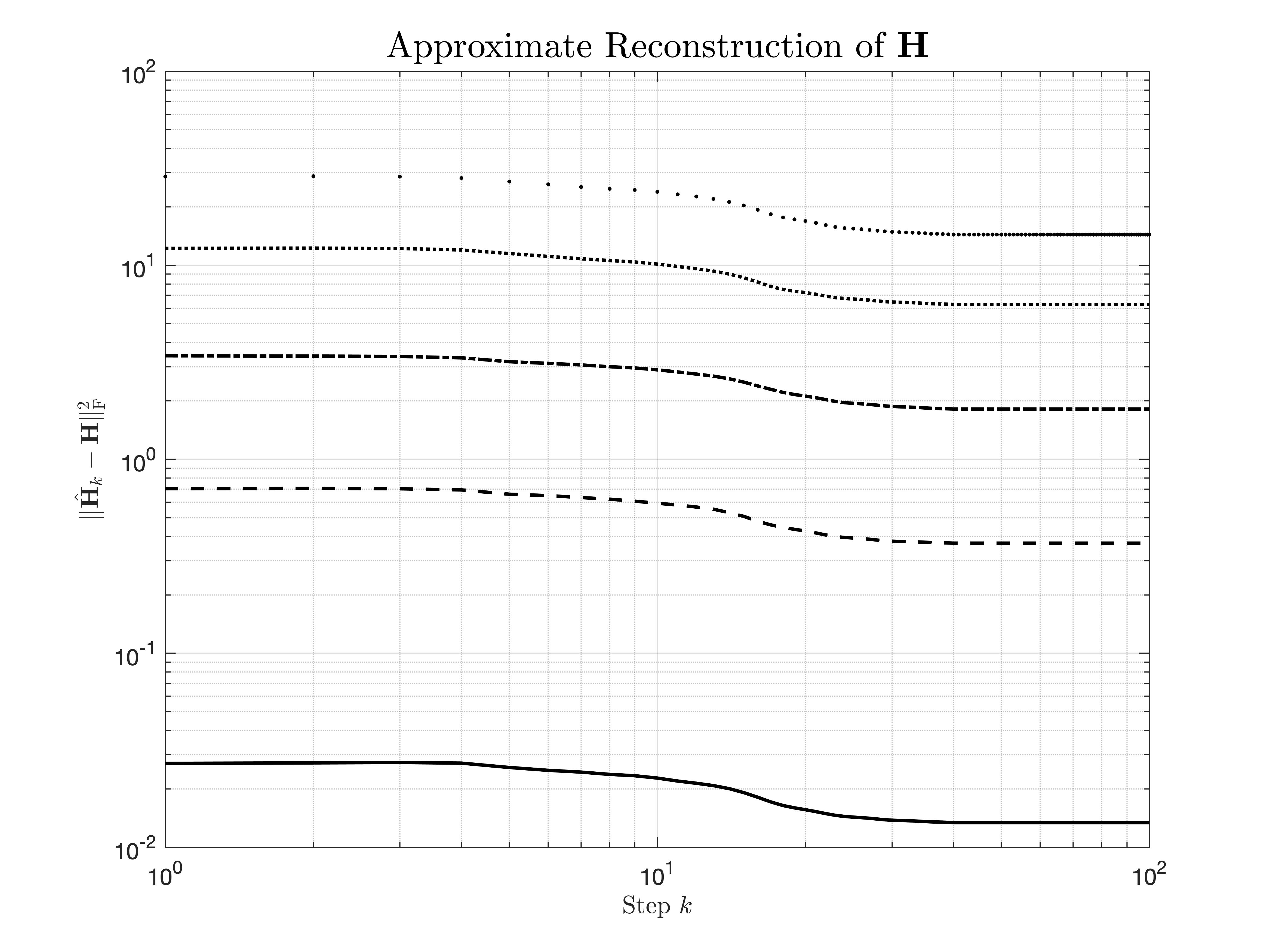}
	\caption{The log--log plot of $\|\bar{\Hb}_k-\Hb\|_{\rm F}^2$ along the interior--point method starting from different initial estimates $\bar{\Hb}_0$. Here $\bar{\Hb}_k$ is the $k$--th iteration of the interior--point method implemented over the problem (\ref{eq:local_optimization}).}
	\label{fig:rec_local}
\end{figure}

\section{Conclusions}\label{sec:conc}
We have shown that distributed computing inevitably led  to  privacy risks in terms the local datasets by both global and local eavesdroppers having access to the entirety of part of the node states. Such threats deserves attention as much of the dynamical states in distributed computing protocols represents physical states, which is  sensitive in terms of both privacy to individuals and vulnerability to eavesdroppers. After presenting a general frame work, we showed that for  network linear equations,   the computational privacy has indeed been completely lost to global eavesdroppers for almost all initial values following existing distributed linear equation solvers, and partially   to local passive or active eavesdroppers  who can monitor or alter one selected node. We proposed a Privacy-Preserving-Summation-Consistent (PPSC) mechanism, and showed that it can be used as  a generic privacy encryption step for  consensus-based distributed computations. The central idea of using the consensus manifold as a place to  hide node privacy while achieving  computation accuracy may be extended to non-consensus based distributed computing tasks in future work.


%

\medskip

\medskip

\medskip

\section*{Appendices}

\medskip

\subsection*{A. Proof of Proposition \ref{prop1}}
(i) Suppose there exist datasets $\mathpzc{D}^\prime\neq\mathpzc{D}$ that yield the same node state trajectory $\mathscr{M}_{\xb_0}(\mathpzc{D})=\mathscr{M}_{\xb_0}(\mathpzc{D}^\prime)$ under $\xb_0$. Due to the existence of $\mathpzc{D}^\prime$, $\mathpzc{D}$ cannot be uniquely determined from $\mathscr{M}_{\xb_0}(\mathpzc{D})$, implying that (\ref{eq:dcp}) is globally computationally private with respect to $\mathpzc{D}$ under $\xb_0$. On the other hand, we suppose that $\mathscr{M}_{\xb_0}(\mathpzc{D})=\mathscr{M}_{\xb_0}(\mathpzc{D}^\prime)$ implies $\mathpzc{D}=\mathpzc{D}^\prime$. As a result, $\mathscr{M}_{\xb_0}$ is injective and $\mathscr{M}_{\xb_0}^{-1}$ exists. Then from $\mathscr{M}_{\xb_0}(\mathpzc{D})$, a unique dataset $\mathpzc{D}$ always be uniquely determined by $\mathscr{M}_{\xb_0}^{-1}$. Therefore, (\ref{eq:dcp}) is not globally computationally private with respect to $\mathpzc{D}$ under $\xb_0$.

\noindent (ii) Along (\ref{eq:dcp}) with randomized update, the observation sequence $\mathscr{M}_{\xb_0}(\mathpzc{D})$ becomes an array of random variables, with its statistical properties fully characterized by its PDF. Define a mapping $\mathpzc{q}$ from the space of $\mathpzc{D}$ to the space of $\pdf\big(\mathscr{M}_{\xb_0}(\mathpzc{D})\big)$. Then $\mathpzc{q}(\mathpzc{D})=\mathpzc{q}(\mathpzc{D}^\prime)$ if and only if $\mathpzc{D}=\mathpzc{D}^\prime$. Thus the proof of (ii) is analogous to (i) and omitted here.
\subsection*{B. Proof of Theorem  \ref{thm1}}
We first provide the following lemma that assists with the proof.
\begin{lemma}\label{lem:reconstructability}
	Consider a linear equation $\eq\in\Eq^\ast_{[nm]}$ and an initial condition $\xb_0\in\R^{nm}$.
	
	(i) Along the CPA (\ref{eq:shi2017}),  $\eq$ is recoverable in the quotient space $\Eq^\ast_{[nm]}/\sim$ by a global dynamics eavesdropper under $\xb_0$  if and only if for each $i\in\mV$, either of the following conditions holds: (a) There exists $s_i\in\Z^{\ge0}$ such that $\PP_i\big(\xb_i(s_i)\big)\neq\xb_i(s_i)$; (b) There exist $r_1^i,\dots,r_m^i\in\Z^{\ge0}$ such that $\xb_i(r_1^i),\dots,\xb_i(r_m^i)$ are affinely independent.
	
	(ii) Along the  the PCA (\ref{eq:PCA}),	$\eq$ is recoverable in the quotient space $\Eq^\ast_{[nm]}/\sim$ by a global dynamics eavesdropper under $\xb_0$ if and only if for each $i\in\mV$, either of the following conditions holds: (a) There exists $s_i\in\Z^{\ge0}$ such that $\sum\limits_{j\in\mN_i}w_{ij}\PP_i\big(\xb_j(s_i)\big)\neq\sum\limits_{j\in\mN_i}w_{ij}\xb_j(s_i)$; (b) There exist $r_1^i,\dots,r_m^i\in\Z^{\ge0}$ such that $\sum\limits_{j\in\mN_i}w_{ij}\xb_j(r_1^i),\dots,\sum\limits_{j\in\mN_i}w_{ij}\xb_j(r_m^i)$ are affinely independent.
\end{lemma}
\begin{proof}
	(i) \noindent({\em Necessity.}) Suppose there exists $\iast\in\mV$ such that neither (a) nor (b) holds. Then $\PP_{\iast}\big(\xb_{\iast}(t)\big)=\xb_{\iast}(t)$ for all $t\in\Z^{\ge0}$.
	%
	As a result, the dynamics (\ref{eq:shi2017}) becomes
	\begin{equation}\notag
	\begin{aligned}
	\xb_{\iast}(t+1) &= \sum\limits_{j\in\mN_{\iast}} w_{ij}\xb_j(t);\\
	\xb_i(t+1) &= \sum\limits_{j\in\mN_i}w_{ij}\xb_j(t)+\alpha\bigg(\PP_i\big(\xb_i(t)\big)-\xb_i(t)\bigg),\ i\neq\iast.
	\end{aligned}
	\end{equation}
	This means that the trajectory of $\xb(t)$ contains no information about $\eq_{\iast}$, except for the knowledge $\PP_{\iast}\big(\xb_{\iast}(t)\big)=\xb_{\iast}(t)$. By the definition of $\PP_i$, the eavesdropper can only attempt to recover $\eq_{\iast}$ from
	\begin{equation}\label{eq:h_iast_set}
	\Hb_{\iast}^\top\xb_{\iast}(t)=\zb_{\iast},\ t\in\Z^{\ge0}.
	\end{equation}
	Since the eavesdropper does not know $\zb_{\iast}$, from the perspective of the eavesdropper who does not know $\zb_{\iast}$, the infinite equation set (\ref{eq:h_iast_set}) is equivalent to
	\begin{equation}\label{eq:h_iast_diff_set}
	\Hb_{\iast}^\top\xb_{\iast}(0)=\Hb_{\iast}^\top\xb_{\iast}(1)=\cdots.
	\end{equation}
	It is clear that (\ref{eq:h_iast_diff_set}) is equivalent to the constraint that $\Hb_{\iast}$ is perpendicular to $\xb_{\iast}(t)-\xb_{\iast}(s)$ for all $t,s\in\Z^{\ge0}$.
	When (b) does not hold, the vectors $\xb_{\iast}(t)$ for all $t\in\Z^{\ge0}$ are affinely dependent, which yields \cite{gallier2011basics} $\dim\big\{\xb_{\iast}(t)-\xb_{\iast}(s): \forall t,s\in\Z^{\ge0}\big\}<m-1$. This implies the solution space of $\Hb_{\iast}$ satisfying (\ref{eq:h_iast_diff_set}) has dimension larger than one. Therefore, $\eq_{\iast}$, and there by $\eq$, is not recoverable. This implies that the CPA (\ref{eq:shi2017}) is computationally private.

	\noindent\big({\em Sufficiency of Condition (a)}\big)
	Suppose the condition (a) holds for any $i\in\mV$. Then according to (\ref{eq:shi2017}), one has
	\begin{equation}\label{eq:global_observability_privacy_condition}
	\xb_i(s_i+1)\neq\sum\limits_{j\in\mN_i}w_{ij}\xb_j(s_i).
	\end{equation}
	Direct calculation shows $\PP_i(\xb)=\bigg(\Ib_m-\frac{\Hb_i\Hb_i^\top}{\Hb_i^\top\Hb_i}\bigg)\xb+\frac{\zb_i\Hb_i}{\Hb_i^\top\Hb_i}$. Based on the dynamics (\ref{eq:shi2017}), we have
	\begin{equation}\label{eq:CPA_proof1}
	\frac{\alpha\Hb_i}{\Hb_i^\top\Hb_i}\big(\zb_i-\Hb_i^\top\xb_i(t)\big)=\xb_i(t+1)-\sum\limits_{j\in\mN_i}w_{ij}\xb_j(t).
	\end{equation}
	Clearly, (\ref{eq:global_observability_privacy_condition}) guarantees that both sides of (\ref{eq:CPA_proof1}) are nonzero when $t=s_i$, which further implies that $\Hb_i$ and $\xb_i(t+1)-\sum\limits_{j\in\mN_i}w_{ij}\xb_j(t)$ are parallel nonzero vectors. Defining
	\begin{equation}\label{eq:def_of_H_hat}
	\hat{\Hb}_i=\frac{\xb_i(s_i+1)-\sum\limits_{j\in\mN_i}w_{ij}\xb_j(s_i)}{\bigg\|\xb_i(s_i+1)-\sum\limits_{j\in\mN_i}w_{ij}\xb_j(s_i)\bigg\|},
	\end{equation}
	there exists $\lambda\neq0$ such that
	\begin{equation}\label{eq:CPA_proof2}
	\Hb_i=\lambda\hat{\Hb}_i.
	\end{equation}
	By replacing $\xb_i(t+1)-\sum\limits_{j\in\mN_i}w_{ij}\xb_j(t)$ and $\Hb_i$ in (\ref{eq:CPA_proof1}) based on the relations given in (\ref{eq:def_of_H_hat}) and (\ref{eq:CPA_proof2}), one has
	\begin{equation}\label{eq:scalar_vector}
	\alpha\hat{\Hb}_i\cdot\bigg(\frac{\zb_i}{\lambda}-\hat{\Hb}_i^\top\xb_i(s_i)\bigg)=\hat{\Hb}_i\cdot\bigg\|\xb_i(s_i+1)-\sum\limits_{j\in\mN_i}w_{ij}\xb_j(s_i)\bigg\|.
	\end{equation}
	It can be observed that both sides of (\ref{eq:scalar_vector}) are a product of a scalar and the nonzero vector $\hat{\Hb}_i$. Thus one can obtain the following equation by letting the scalars at the left-- and right--hand side be equal.
	\begin{equation}\label{eq:CPA_proof3}
	\alpha\bigg(\frac{\zb_i}{\lambda}-\hat{\Hb}_i^\top\xb_i(s_i)\bigg)=\bigg\|\xb_i(s_i+1)-\sum\limits_{j\in\mN_i}w_{ij}\xb_j(s_i)\bigg\|.
	\end{equation}
	By letting $\lambda=1$ in (\ref{eq:CPA_proof3}), one can immediately recover the following linear equation.
	\begin{equation}\label{eq:linear_equation_rec_results}
	\begin{aligned}
	&\hat{\eq}_i:\quad\bigg(\xb_i(s_i+1)-\sum\limits_{j\in\mN_i}w_{ij}\xb_j(s_i)\bigg)^\top\yb=\bigg(\xb_i(s_i+1)-\sum\limits_{j\in\mN_i}w_{ij}\xb_j(s_i)\bigg)^\top\xb_i(s_i)\\
	&+\frac{1}{\alpha}\bigg\|\xb_i(s_i+1)-\sum\limits_{j\in\mN_i}w_{ij}\xb_j(s_i)\bigg\|^2.
	\end{aligned}
	\end{equation}
	It can be easily verified $\hat{\eq}_i\sim\eq_i$.
	
	\noindent\big({\em Sufficiency of Condition (b)}\big)
	Suppose the condition (b) holds for any $i\in\mV$. We have shown above the linear equation $\eq$ is recoverable if the condition (a) holds, and thus in this case we suppose the condition (a) does not hold, namely $\PP_i\big(\xb_i(t)\big)=\xb_i(t)$ for all $t$. By the definition of $\PP_i$, we know
	\begin{equation}\label{eq:h_i_x_i}
	\Hb_i^\top\xb_i(r_k^i)=\zb_i,\ k=1,\dots,m.
	\end{equation}
	Further, it follows from (\ref{eq:h_i_x_i})
	\begin{equation}\notag
	\Hb_i^\top\big(\xb_i(r_{k+1}^i)-\xb_i(r_k^i)\big)=0,\ k=1,\dots,m-1,
	\end{equation}
	which can be written into a compact form $
	\Ub\Hb_i=0,
	$
	where $\Ub$ is an $m-1$--by--$m$ matrix with each row being $\Ub_k=\xb_i(r_{k+1}^i)-\xb_i(r_k^i),\ k=1,\dots,m-1$. According to \cite{gallier2011basics},
	$\rank(\Ub)=m-1$ and thereby the dimension of the kernel of $\Ub$ is one. We can pick any $\Hb_i^\prime\neq0$ in the kernel of $\Ub$ and compute $\zb_i^\prime=\Hb_i^{\prime\top}\xb_i(r_k^i)$ with any $k$. Finally, $\eq_i^\prime:\Hb_i^{\prime\top}\yb=\zb_i^\prime$ that is equivalent to $\eq_i$ can be recovered.
	
	We have shown above that if either the condition (a) or the condition (b) holds for any $i$, then $\eq_i^\prime\sim\eq_i$ is recoverable.
	
	\noindent (ii) For all $i\in\mV$, there holds according to (\ref{eq:PCA})
	\begin{equation}\label{eq:PCA_proof1}
	\frac{\Hb_i}{\Hb_i^\top\Hb_i}\big(\zb_i-\Hb_i^\top\sum\limits_{j\in\mN_i}w_{ij}\xb_j(t_i)\big)=\xb_i(t_i+1)-\sum\limits_{j\in\mN_i}w_{ij}\xb_j(t_i).
	\end{equation}
	It is worth noting that (\ref{eq:PCA_proof1}) along the PCA (\ref{eq:PCA}) is similar to (\ref{eq:CPA_proof1}) along the CPA (\ref{eq:shi2017}). By the same arguments as the proof for the CPA (\ref{eq:shi2017}), we can establish the desired statement as well.
\end{proof}
\noindent (i) Based on Lemma \ref{lem:reconstructability}, we know if $\eq$ is recoverable under $\xb(0)$ along the CPA (\ref{eq:shi2017}), then there exists $\iast\in\mV$ such that neither the condition (a) nor the condition (b) holds.
Correspondingly, we define sets
\begin{align*}
\mathrm{S}_t&=\big\{\xb(0)\in\R^{nm}:\PP_i\big(\xb_{\iast}(t)\big)=\xb_{\iast}(t),\textnormal{ where }\xb_{\iast}(t)\textnormal{ evolves according to }(\ref{eq:shi2017})\big\},\ t\in\Z^{\ge0}\\
\mathrm{T}&=\big\{\xb(0)\in\R^{nm}:\xb_{\iast}(t),t\in\Z^{\ge0}\textnormal{ are affinely dependent,}\textnormal{ where }\xb_{\iast}(t)\textnormal{ evolves according to }(\ref{eq:shi2017})\big\}.
\end{align*}
According to the claim above, $\xb(0)\in\bigcap_{t=0}^\infty\mathrm{S}_t\mcap\mathrm{T}$. Clearly, $\xb(0)\in\mathrm{S}_0$ says
$
\frac{\Hb_{\iast}}{\Hb_{\iast}^\top\Hb_{\iast}}\big(\zb_{\iast}-\Hb_{\iast}^\top\xb_{\iast}(0)\big)=0.
$
Recall that $\eq_i$ for all $i$ are assumed to be nontrivial. Then the subspace $\mathrm{S}_0$ has dimension $nm-1$, and thus has measure zero in $\R^{nm}$. Since $\bigcap_{t=0}^\infty\mathrm{S}_t\mcap\mathrm{T}$ is a subset of $\mathrm{S}_0$, it also has measure zero. This finishes proving that the CPA (\ref{eq:shi2017}) is not computationally private under almost all initial conditions.

\medskip

\noindent (ii) Analogously, if $\eq$ is recoverable under $\xb(0)$ along the PCA (\ref{eq:PCA}), then based on Lemma \ref{lem:reconstructability}, there necessarily exists $\iast\in\mV$ such that
$$\frac{\Hb_{\iast}^\top}{\Hb_{\iast}^\top\Hb_{\iast}}\sum\limits_{j\in\mN_i}w_{\iast j}\xb_j(0)=\zb_{\iast},$$
which restricts the dimension of the subspace of $\xb(0)$ to $nm-1$. This leads to the same conclusion as the CPA (\ref{eq:shi2017}).

\subsection*{C. Proof of Theorem \ref{thm:sensitivity_dp_CPA}}

The DP--DLES can be written as
\begin{equation}\label{eq:differential_dynamics_CPA}
\begin{aligned}
\xbf_i(t) &= \PPo\big(\xb_i(t)\big)\\
\xbs_i(t) &= \xbf(t)+\omegab_i(t)\\
\xb_i(t+1) &= \sum\limits_{j\in\mN_i}\Wb_{ij}\xbs_j(t)+\alpha(t)\big(\PP_i\big(\xbf_i(t)\big)-\xbf_i(t)\big).
\end{aligned}
\end{equation}
Define $\PPo^\ast\big(\xb(t)\big)=[\PPo\big(\xb_1(t)\big)^\top\ \dots\ \PPo\big(\xb_n(t)\big)^\top]^\top$ and $\omegab(t)=[\omegab_1(t)^\top\ \dots\ \omegab_n(t)^\top]^\top$. By removing the intermediate states $\xbf_i(t)$ and $\xbs_i(t)$, we can rewrite (\ref{eq:differential_dynamics_CPA}) compactly as
\begin{equation}\label{eq:differential_dynamics_compact_CPA}
\xb(t+1) = \big(\Wb\otimes\Ib_m\big)\big(\xb(t)+\omegab(t)\big)+\alpha(t)\big(\tz -\tH \PPo^\ast\big(\xb(t)\big)\big),
\end{equation}
where $\tH=\diag\bigg(\frac{\Hb_1\Hb_1^\top}{\Hb_1^\top\Hb_1},\dots,\frac{\Hb_n\Hb_n^\top}{\Hb_n^\top\Hb_n}\bigg),\tz=\begin{bmatrix}
\frac{\zb_1\Hb_1}{\Hb_1^\top\Hb_1}& \dots & \frac{\zb_n\Hb_n}{\Hb_n^\top\Hb_n}\end{bmatrix}^\top$. We now associate each iteration of DP--DLES at time $t\in\Z^{\ge0}$ with a mechanism $\mM_t:\Eq^\ast_{[nm]}\times\R^{nm}\to\R^{nm}$ satisfying $\mM_t\big(\eq,\xb(t)\big)=\xb(t+1)$. Consider two adjacent linear equations $\eq:\Hb\yb=\zb,\eq^\prime=\Hb^\prime\yb=\zb^\prime\in\Eq^\ast_{[nm]}$. We define $\tH ^\prime,\tz ^\prime$ for $\eq^\prime$ similarly as $\tH ,\tz $ for $\eq$. Then for $\eq,\eq^\prime$, there holds
\begin{align}
&\quad\frac{\Pr\big(\mM_t(\eq,\xb(t))=\xb(t+1)\big)}{\Pr\big(\mM_t(\eq^\prime,\xb(t))=\xb(t+1)\big)}\notag\\
&\overset{\rm a)}{=}\frac{{\pdf\bigg((\Wb\otimes\Ib_m)^{-1}\bigg(\xb(t+1)}{-\alpha(t)\big(\tz -\tH \PPo^\ast\big(\xb(t)\big)\big)\bigg)-\xb(t)\bigg)}}{{\pdf\bigg((\Wb\otimes\Ib_m)^{-1}\big(\xb(t+1)}{-\alpha(t)\big(\tz ^\prime-\tH ^\prime\PPo^\ast\big(\xb(t)\big)\big)\bigg)-\xb(t)\bigg)}}\notag\\
&\le\exp\bigg(\frac{\alpha(t)}{c\phi^t}\big\|(\Wb\otimes\Ib_m)^{-1}\big((\tH -\tH ^\prime)\PPo^\ast\big(\xb(t)\big)-(\tz -\tz ^\prime)\big)\big\|_1\bigg)\notag\\
&\le\exp\bigg(\frac{\alpha(t)}{c\phi^t}\big\|\Wb^{-1}\otimes\Ib_m\big\|_1\big\|(\tH -\tH ^\prime)\PPo^\ast\big(\xb(t)\big)-(\tz -\tz ^\prime)\big\|_1\bigg)\notag\\
&=\exp\bigg(\frac{\alpha(t)}{c\phi^t}\big\|\Wb^{-1}\big\|_1\sum\limits_{i=1}^n\bigg\|\bigg(\frac{\Hb_i\Hb_i^\top}{\Hb_i^\top\Hb_i}-\frac{\Hb_i^\prime\Hb_i^{\prime\top}}{\Hb_i^{\prime\top}\Hb_i^\prime}\bigg)\PPo\big(\xb_i(t)\big)-\bigg(\frac{\zb_i\Hb_i}{\Hb_i^\top\Hb_i}-\frac{\zb_i^\prime\Hb_i^\prime}{\Hb_i^{\prime\top}\Hb_i^\prime}\bigg)\bigg\|_1\bigg)\notag\\
&\overset{\rm b)}{\le}\exp\bigg(\frac{\alpha(t)\sqrt{nm}}{c\phi^t}\big\|\Wb^{-1}\big\|\sum\limits_{i=1}^n\bigg(\bigg\|\frac{\Hb_i\Hb_i^\top}{\Hb_i^\top\Hb_i}-\frac{\Hb_i^\prime\Hb_i^{\prime\top}}{\Hb_i^{\prime\top}\Hb_i^\prime}\bigg\|\big\|\PPo\big(\xb_i(t)\big)\big\|+\bigg\|\frac{\zb_i\Hb_i}{\Hb_i^\top\Hb_i}-\frac{\zb_i^\prime\Hb_i^\prime}{\Hb_i^{\prime\top}\Hb_i^\prime}\bigg\|\bigg)\bigg)\notag\\
&\overset{\rm c}{\le}\exp\bigg(\frac{\lambda\psi^t\sqrt{nm}}{c\phi^t\sigm(\Wb)}\big(B\dH+\dz\big)\bigg),\label{eq:mM_t_inequality}
\end{align}
where a) uses (\ref{eq:differential_dynamics_compact_CPA}) and the PDF of Laplace random variables, b) is obtained from norm  inequalities, and c) comes from the linear--equation adjacency. By omitting the notation $\eq$, the following composition relation holds $\mM=(\mM_0,\mM_1\circ\mM_0,\dots,\dots\circ\mM_1\circ\mM_0)$. To let $\mM$ be $\big(\epsilon,\dH,\dz\big)$--differentially private, based on (\ref{eq:mM_t_inequality}), the following must hold from the composition property of differential privacy (see, e.g., Theorem 1 of \cite{mcsherry2009privacy}):
\begin{equation}\notag
\sum\limits_{t=0}^\infty\frac{\lambda\psi^t\sqrt{nm}}{c\phi^t\sigm(\Wb)}\big(B\dH+\dz\big) \le \epsilon,
\end{equation}
which further implies the desired result.


%

\subsection*{D. Proof of Theorem  \ref{thm2}}
We divide the proof into three steps.

\noindent Step 1. The CPA (\ref{eq:shi2017}) can be compactly rewritten as
\begin{equation}\label{eq:compact1}
\xb(t+1) = \Fb\xb(t)+\alpha\tz,
\end{equation}
where $\Fb=\Wb\otimes\Ib_m-\alpha\diag\bigg(\frac{\Hb_1\Hb_1^\top}{\Hb_1^\top\Hb_1},\dots,\frac{\Hb_n\Hb_n^\top}{\Hb_n^\top\Hb_n}\bigg)$ and $\tz ^\top=\begin{bmatrix}
\frac{\zb_1\Hb_1^\top}{\Hb_1^\top\Hb_1} & \dots & \frac{\zb_n\Hb_n^\top}{\Hb_n^\top\Hb_n}
\end{bmatrix}$.
Any passive local dynamics eavesdropper $i\in\mV$ knows a solution $\yb^\ast$ to the linear equation $\eq$ as a result of implementing the CPA (\ref{eq:shi2017}). By introducing $\gammab(t)=\xb(t)-\1\otimes\yb^\ast$, the eavesdropper can rewrite (\ref{eq:compact1}) into
\begin{align}
\gammab(t+1) = \Fb\gammab(t).\label{eq:compact1_rewrite}
\end{align}
Since the eavesdropper $i$ has access to only the state observations of the nodes in $\mN_i$, its observation can be described by
\begin{equation}\notag
\vb(t) = (\Eb_i\otimes\Ib_m)\gammab(t)+\1\otimes\yb^\ast.
\end{equation}
Alternatively, knowing $\yb^\ast$, this passive eavesdropper has access to
\begin{equation}\label{eq:compact_output}
\yb(t)=\vb(t)-\1\otimes\yb^\ast=(\Eb_i\otimes\Ib_m)\gammab(t).
\end{equation}
\noindent Step 2. We apply the system identification approach  from \cite{budin1971minimal} to the system described by (\ref{eq:compact1_rewrite})--(\ref{eq:compact_output}) with state $\gammab(t)$ and output $\yb(t)$. According to \cite{budin1971minimal}, if $(\Fb,\Eb\otimes\Ib_m)$ is completely observable, i.e., the matrix
\begin{align}
\mO=\begin{bmatrix}
\Eb_i\otimes\Ib_m\\
(\Eb_i\otimes\Ib_m)\Fb\\
\vdots\\
(\Eb_i\otimes\Ib_m)\Fb^{nm-\left|\mN_i\right|m}
\end{bmatrix}\notag
\end{align}
has full column rank, and there exist $t_1<\dots<t_{nm}$ such that the vectors $\xb(t_k),\ k=1,\dots,nm$ are linearly independent, then the passive local dynamics eavesdropper $i$ can identify $(\Fb_\ast ,\Cb_\ast )\in\R^{nm\times nm}\times\R^{\left|\mN_i\right|m\times nm}$ from $\big(\yb(t)\big)_{t=0}^\infty$ by the following procedure.
\begin{enumerate}[(i)]
	\item Collect a number of consecutive outputs of the system (\ref{eq:compact1_rewrite})--(\ref{eq:compact_output}) $\big(\yb(t)\big)_{t=t_k}^{t_k+nm-\left|\mN_i\right|m}$ for $k=1,\dots,nm$ and construct
	\begin{align*}
	\uY&=\begin{bmatrix}
	\yb(t_1) & \cdots & \yb(t_{nm})\\
	\vdots & \ddots & \vdots\\
	\yb(t_1+nm-\left|\mN_i\right|m-1) & \cdots & \yb(t_{nm}+nm-\left|\mN_i\right|m-1)
	\end{bmatrix},\\
	\oY&= \begin{bmatrix}
	\yb(t_1+1) & \cdots & \yb(t_{nm}+1)\\
	\vdots & \ddots & \vdots\\
	\yb(t_1+nm-\left|\mN_i\right|m) & \cdots & \yb(t_{nm}+nm-\left|\mN_i\right|m)
	\end{bmatrix}.
	\end{align*}
	
	Since $\uY = \mO[
	\xb(t_1)\ \cdots\ \xb(t_{nm})]$ and $\mO$ has full column rank, $\rank(Y)=\rank
	[\xb(t_1)\ \cdots\ \xb(t_{nm})]=nm$.
	\item Write the state--output relation for the collected output data as
	\begin{equation}\label{eq:state-output}
	\begin{aligned}
	\oY &= \mO(\Wb\otimes\Ib_m-\alpha\tH )\begin{bmatrix}
	\xb(t_1) & \cdots & \xb(t_{nm})
	\end{bmatrix},\\
	\uY &= \mO\begin{bmatrix}
	\xb(t_1) & \cdots & \xb(t_{nm})
	\end{bmatrix}.
	\end{aligned}
	\end{equation}
	\item Choose a row--selecting matrix $\Sb:\R^{nm\times \left|\mN_i\right|(nm-\left|\mN_i\right|m)}$ such that $\Sb\uY$ is invertible, and further, based on (\ref{eq:state-output}), obtains $\Fb_\ast :=\Sb\mO(\Wb\otimes\Ib_m-\alpha\tH )(\Sb\mO)^{-1}=\Sb\oY(\Sb\uY)^{-1}$ and then $\Cb_\ast =[\Ib_{m\left|\mN_i\right|}\ 0]$.
\end{enumerate}
The obtained pair $(\Ab_\ast,\Cb_\ast )$ is different from $(\Fb,\Eb_i\otimes\Ib_m)$ subject to a coordinate change based on Proposition 2.3 of \cite{de2000minimal}:
\begin{align}
\Tb^{-1}\Fb\Tb = \Fb_\ast ,\ (\Eb_i\otimes\Ib_m)\Tb=\Cb_\ast \label{eq:realization_A_C}
\end{align}
with an invertible $\Tb\in\R^{nm\times nm}$.

\noindent Step 3. We now come back to the point of view of the eavesdropper. Since the local dynamics eavesdropper $i$ knows $\Eb_i$, it can determine the set $\mathcal{T}$ based on (\ref{eq:realization_A_C}). By (\ref{eq:realization_A_C}), there exists some $\Tb\in\mathcal{T}$ such that $\diag\bigg(\frac{\Hb_1\Hb_1^\top}{\Hb_1^\top\Hb_1},\dots,\frac{\Hb_n\Hb_n^\top}{\Hb_n^\top\Hb_n}\bigg)=(\Wb\otimes\Ib-\Tb^{-1}\Fb_\ast \Tb)/\alpha$. It is worth noting that the eavesdropper can compute $\zb=\Hb\lim\limits_{t\to\infty}\xb_i(t)$.	This completes the proof.
\subsection*{E. Proof of Theorem  \ref{thm3}}

We first  establish a lemma that assists with the presentation of the privacy of CPA (\ref{eq:shi2017}) against the active local eavesdropper $i$.
\begin{lemma}\label{lem:stability}
	The eigenvalues of $\Wb\otimes\Ib_m-\alpha\tH$ for $\alpha\in\big(0,\lam(\Wb)+1\big)$ are all strictly less than one in absolute value if $\eq\in\Eq_{[nm]}^\ast$ has a unique solution.
\end{lemma}
\begin{proof}
	By the randomness of $\Wb$ and Lemma 8.1.21 \cite{horn1990matrix}, for any $\vb=[\vb_1^\top\ \dots\ \vb_n^\top]^\top\neq0$ with $\vb_i\in\R^{m}$, there holds $\lam(\Wb)\|\vb\|^2\le \vb^\top(\Wb\otimes\Ib_m)\vb\le\|\vb\|^2$, where the right--hand equality holds if and only if $\vb_1=\dots=\vb_n$. In addition, $0\le \vb^\top\tH\vb=\sum\limits_{i=1}^n\frac{\|\Hb_i^\top\vb_i\|^2}{\Hb_i^\top\Hb_i}\le \|\vb\|^2$ for all $\vb$, where the left--hand equality holds if and only if each $\vb_i$ is perpendicular to $\Hb_i$. Then it can be concluded that $\big(\lam(\Wb)-\alpha\big)\|\vb\|^2\le\vb^\top(\Wb\otimes\Ib_m-\alpha\tH)\vb\le\|\vb\|^2$ for all $\vb$, where the right--hand equality holds if and only if $\vb_1=\dots=\vb_n\in\ker(\Hb)$. Since $\alpha<\lam(\Wb)+1$  and $\ker(\Hb)=\emptyset$, $-\|\vb\|^2<\vb^\top(\Wb\otimes\Ib_m-\alpha\tH)\vb<\|\vb\|^2$, implying that all its eigenvalues fall in $(-1,1)$.
\end{proof}

Again we divide the proof into three steps.

\noindent Step 1. Let a $T$-periodic signal $\rb:\Z^{\ge0}\to\R^m$ be added to $\xb_i(t)$ at each time $t$ by the eavesdropper $i$ with $T\ge 2nm+1$ and
\begin{equation}\notag
\Rb=\begin{bmatrix}
\rb(0) & \rb(1) & \cdots & \rb(T-1)\\
\rb(T-1) & \rb(0) & \cdots & \rb(T-2)\\
\vdots & \vdots & \ddots & \vdots\\
\rb(1) & \rb(2) & \cdots & \rb(0)
\end{bmatrix}
\end{equation}
having full rank.
Now the system (\ref{eq:compact1_rewrite}) becomes
\begin{equation}\label{eq:compact1_frequency}
\gammab(t+1) = \Fb\gammab(t)+\Bb\rb(t),\ t>t^\ast,
\end{equation}
where $\Bb\in\R^{nm}$ is the Kronecker product of an all--zeros vector except for the $i$--the component being one and the $m$--by--$m$ identity matrix. From Lemma \ref{lem:stability}, $\Fb$ is a stable matrix if $\eq$ has a unique solution.

\noindent Step 2. We apply the system identification approach  from \cite{mckelvey1995subspace} to the system (\ref{eq:compact_output})--(\ref{eq:compact1_frequency}) with state $\gammab(t)$ and output $\yb(t)$ by the following procedure:
\begin{enumerate}[(i)]
	\item Collect $\big(\yb(t)\big)_{t=kT}^{kT+T-1}$ for some large $k$, namely consecutive outputs of the system (\ref{eq:compact_output})--(\ref{eq:compact1_frequency}) within one period $T$.
	\item Compute
	\begin{align*}\notag
	\sum\limits_{j=0}^\infty
	\begin{bmatrix}
	\Gb_{jT} & \dots & \Gb_{jT+T-1}
	\end{bmatrix}=
	\begin{bmatrix}
	\yb(kT) & \dots & \yb(kT+T-1)
	\end{bmatrix}
	\cdot\Rb^{-1},
	\end{align*}
	where $\Gb_j=(\Eb_i\otimes\Ib_m)(\Wb\otimes\Ib_m-\alpha\tH )^{j-1}\bb$ and $\Gb_0=0$ are the impulse responses of the system (\ref{eq:compact_output})--(\ref{eq:compact1_frequency}). The result is finite because the system (\ref{eq:compact1_frequency}) is stable by Lemma \ref{lem:stability}.
	
	\item Reorganize $\sum\limits_{j=0}^\infty\Gb_{jT+l}$ with $l=0,\dots,T-1$ in a Hankel matrix $\Mb\in\R^{p\times q}$
	\begin{equation}\notag
	\Mb =\begin{bmatrix}
	\sum\limits_{j=0}^\infty\Gb_{jT+1} & \cdots & \sum\limits_{j=0}^\infty\Gb_{jT+q}\\
	\vdots & \ddots & \vdots\\
	\sum\limits_{j=0}^\infty\Gb_{jT+p} & \cdots & \sum\limits_{j=0}^\infty\Gb_{jT+p+q-1}
	\end{bmatrix}
	\end{equation}
	with $p>nm,p+q=T$ and computes its singular value decomposition as
	$$\Mb=\begin{bmatrix}
	\Ub_s & \Ub_o
	\end{bmatrix}
	\begin{bmatrix}
	\Sigmab_s & 0\\
	0 & \Sigmab_o
	\end{bmatrix}
	\begin{bmatrix}
	\Vb_s^\top\\
	\Vb_o^\top
	\end{bmatrix},\ \Sigmab_s\in\R^{nm\times nm}.$$
	Then $\Ab_\ast=([\Ib_{\left|\mN_i\right|m(p-1)}\ 0]\Ub_s)^+[0\ \Ib_{\left|\mN_i\right|m(p-1)}]\Ub_s$ and $\Cb_\ast =[\Ib_{\left|\mN_i\right|m}\ 0]\Ub_s$, where $\Tb^{-1}(\Wb\otimes\Ib_m-\alpha\tH )\Tb = \Ab_\ast,\ (\Eb_i\otimes\Ib_m)\Tb=\Cb_\ast $ and $\Tb\in\R^{nm\times nm}$ is invertible.

\end{enumerate}
According to Theorem 1 in \cite{mckelvey1995subspace}, with $\Fb$ being stable, the obtained pair $(\Fb_\ast ,\Cb_\ast )$ and $(\Fb,\Cb)$ also satisfy the relation in (\ref{eq:realization_A_C}).

\noindent Step 3. From (\ref{eq:realization_A_C}), the eavesdropper can determine the set $\mathcal{T}$. By (\ref{eq:realization_A_C}), there exists some $\Tb\in\mathcal{T}$ such that $\diag\bigg(\frac{\Hb_1\Hb_1^\top}{\Hb_1^\top\Hb_1},\dots,\frac{\Hb_n\Hb_n^\top}{\Hb_n^\top\Hb_n}\bigg)=(\Wb\otimes\Ib-\Tb^{-1}\Fb_\ast \Tb)/\alpha$. Further, the eavesdropper can compute $\zb=\Hb\yb^\ast$.	This completes the proof.




%

\end{document}